\documentclass[cmp]{svjour2}

\usepackage{amsmath,amsfonts,amssymb,graphics,graphicx,epsfig,times,color,bbm,tikz}

\usepackage[square,sort&compress,comma,numbers]{natbib}

\newcommand{\cc}{\ensuremath{\mathbbm{C}}}
\newcommand{\zz}{\ensuremath{\mathbbm{Z}}}
\newcommand{\nn}{\ensuremath{\mathbbm{N}}}
\newcommand{\rr}{\ensuremath{\mathbbm{R}}}

\newcommand{\torus}{\ensuremath{\mathbb{T}}}

\newcommand{\Ei}{{\rm Ei}}
\newcommand{\sech}{{\rm sech}}
\newcommand{\arsinh}{{\rm arsinh}}

\newcommand{\jmod}{\,{\rm mod}}
\newcommand{\me}{\ensuremath{\mathrm{e}}}
\newcommand{\mi}{\ensuremath{\mathrm{i}}}

\DeclareMathOperator{\wind}{wind}
\DeclareMathOperator{\bild}{range}

\newcommand{\id}{\ensuremath{\mathbbm{1}}}

\newcommand{\je}[1]{\textcolor{black}{#1}}
\newcommand{\hb}[1]{\textcolor{black}{#1}}

\DeclareMathOperator{\tr}{tr}

\begin{document}

\title{Mutual information area laws for thermal free fermions}

\author{H.\ Bernigau\inst{1}, M.\ J.\ Kastoryano\inst{2}, and J.\ Eisert\inst{2}}

\institute{1 Max Planck Institute for Mathematics in the Sciences, 04103 Leipzig, Germany\\
2 Dahlem Center for Complex Quantum Systems, Freie Universit{\"a}t Berlin, 14195 Berlin, Germany}

\date{\today}

  \maketitle
 \begin{abstract}
We provide a rigorous and asymptotically exact expression 
of the mutual information of translationally invariant free 
fermionic lattice systems in a Gibbs state. In order to arrive at this result, we introduce a novel framework
for computing determinants of T\"oplitz operators with smooth symbols, and
for treating T\"oplitz matrices with system size dependent entries. The 
asymptotically exact mutual information for a partition of the one-dimensional lattice
satisfies an area law, with a prefactor which we compute explicitly.
As examples, we discuss the fermionic  XX model in one dimension and free fermionic models on the 
torus in higher dimensions in detail. Special emphasis is put onto the discussion of 
the temperature dependence of the mutual information, 
scaling like the logarithm of the inverse temperature, hence confirming an expression suggested by
conformal field theory. We also comment on the applicability of the formalism to treat open systems driven by
quantum noise. 
In the appendix, we derive useful bounds to the mutual information in terms of
purities. Finally, we provide a detailed error analysis for finite system sizes. This analysis 
is valuable in its own right for the abstract theory of T\"oplitz determinants.

\end{abstract}

\tableofcontents

\section{Introduction}

How do the  correlations of natural quantum states of many-body systems behave? If ``natural'' 
is taken to mean ``generic'' 
in the sense of a random pure state drawn from the Haar measure, 
then the answer to this question is: Subsystems will almost surely be very nearly maximally correlated with their
complementary subsystems. However, this is 
 not the situation that one is usually interested in in many-body and condensed-matter physics. Ground
states of local Hamiltonians typically exhibit far less entanglement 
than that suggested by the previous argument.
Indeed, there is a large body of evidence 
\cite{Area,Area2,Peschel,Bombelli,Srednicki,Wilczek,Harmonic,Latorre,Cardy,Jin,GraphStates,Area1,Zanardi,MMW,Klitch,NumericsFermions,Toeplitz2,Toeplitz3,AreaLong,Area3,Fermi1,Fermi3,Fermi4,Hastings,HalfSpaces,FH,UniversalKorepin,Toeplitz1,OurPRL,Detectability,Brandao,Micha,Bruno}, suggesting 
that ground states of gapped quantum many-body systems on lattices satisfy an \textit{area law} \cite{Area}
for the entanglement entropy. In other words, given the pure state of a lattice system, if one distinguishes a certain (connected) region of the lattice, then 
the von-Neumann (or  Renyi) entropy of the reduced state associated with this region 
does not grow, as one might expect, like the  number of degrees of freedom of this region (its "volume"). Instead, it  scales 
like the number of degrees of freedom on the boundary, hence as its ``area''. 
Area laws have been shown for all gapped one-dimensional models with local interactions \cite{Hastings,Detectability,Brandao}.  
Extensions to classes of higher-dimensional lattice models have also been obtained \cite{Area1,Area2,OurPRL}. 
These results and others therefore suggest that ground states of quantum many-body systems
are much less entangled than they could be. 
These measures of entanglement -- refined quantities revealing much 
more detailed information about the structure of correlations
than more conventional correlation functions -- in a way inherit the
decay of correlations. 
This deep insight is also at the basis of classical efficient simulations
of quantum many-body systems:
The formulation and analysis of many-body systems in terms of matrix-product and other tensor network states have put this intuition on a solid theoretical  footing \cite{Scholl,2d,Area}.
In numerical simulations of many-body systems, the importance of such approximations can hardly be overestimated. As a consequence, the last decade has seen an enormous amount of  
interest in the study of entanglement in quantum many-body systems in the 
condensed matter context and in issues of numerical simulation \cite{Scholl,2d,Area,Peschel3}, 
entanglement spectra \cite{Peschel3,Peschel2,Spectrum,Spectrum2,Spectrum3,Boundary}, and their
relationship to quantifiers of topological order \cite{Spectrum,Spectrum2,Spectrum3,Boundary}.

Can this analysis be extended to thermal states? A moment of thought reveals that the entanglement entropy should in fact satisfy a volume law for thermal states. However, it turns out that the entanglement entropy is not the right quantity to grasp 
the correlations of a mixed quantum state.  There exist many such measures, but the most natural and the most frequently adopted one 
is the {\it mutual information}. For 
a subset $A$ of sites of the lattice and its complement $B$, 
it is defined as
\begin{equation}
	I(A:B)= S(\rho_A)+ S(\rho_B) - S(\rho),
\end{equation}
where $\rho_A$ and $\rho_B$ are the reduced states of $\rho$ with respect to $A$ and $B$, respectively. For pure states, the mutual information reduces to twice the entanglement entropy. One might then hope that this quantity fulfills an area law for thermal states of local hamiltonians on a lattice.
This turns out to be true in fact for any fixed temperature \cite{Mutual,BR}. Specifically, for
the Gibbs states of a Hamiltonians with local interactions on a lattice, and with bounded operator norm $\|h\|$, one finds that
\begin{equation}
\label{linearAreaLaw}
	I(A:B)\leq 2 \beta \|h\| \,\, 
	|\partial A|,
\end{equation}
where $\partial A$ denotes the boundary area of the region labelled $A$. This statement is in fact
surprisingly simple to prove.  It following rather directly from the extremality of Gibbs states with respect to the 
free energy.

The strength of this result - that it is completely general -  also constitutes its weakness; Eq. (\ref{linearAreaLaw}) does not say anything about the correlation behaviour of thermal states of specific models.  Given that the bound is linearly divergent in the inverse temperature, the law becomes less and less tight in the limit of small temperatures.  In particular, at zero temperature, we know that there can be logarithmic corrections to the area law for critical systems. Hence, it would be desirable to have bounds which depend more explicitly on the properties of the system at hand. What is more, asymptotically exact results for important classes of ``laboratory'' models seem important.

In this work, we present asymptotically exact results for important classes of ``laboratory'' models.
We introduce a framework capable of rigorously computing 
the mutual information in free fermionic lattice models. This framework is based on new approximation techniques as well as on  proof tools for dealing with T{\"o}plitz matrices with smooth symbols.
We separate the presentation in a discussion of the one-dimensional case and one of the 
higher-dimensional situation of a cubic lattice with the topology of a torus. We find rigorous
area laws for the mutual information
that are meaningful in the limit of small temperatures, and converge to the known results for the entanglement entropy for ground states \cite{FH,Jin}. This feature is even present in higher-dimensional fermionic lattice models  \cite{Fermi1,Fermi3,Fermi4,HalfSpaces,Klitch}. \\
For \hb{the critical model at} low temperatures, we \hb{bound the mutual information by} a logarithmic divergent term in the inverse temperature, in fact confirming predictions of conformal field theory. This scaling suggests 
that in order to detect criticality, extremely low temperatures are
needed. This is an exponential improvement in the inverse temperature to the above bound. \hb{For the low temperature limit of the non-critical model we can bound the mutual information by a term that is exponentially decaying in the inverse temperature. The  best decaying rate in our estimate can be bounded by the energy gap. Finally for the high temperature limit we prove an asymptotic decay of the mutual information  proportional to the inverse square of the temperature.}

Mathematically, we discuss new methods of approximation, allowing to use T{\"o}plitz matrix
techniques to studying arbitrary subsystems of translationally invariant fermionic models. What is more,
we introduce a novel second order Szeg{\"o} theorem, simplifying Widom's formula for smooth symbols. This approach allows to use the trace formula without encountering an infinite sum leading to arbitrarily highly
oscillatory terms. Hence, apart from the application in the study of quantum many-body physics, we expect our results to be a 
highly 
useful tool in the abstract theory of T\"oplitz determinants. The paper is structured as follows:
\begin{itemize}
\item In Section \ref{Stage}, we introduce Renyi entropy mutual informations for free fermionic models, discuss fermionic
covariance matrices and present new, easily computable 
and practical upper and lower bounds to mutual informations.
\item 
In Section \ref{Statement}, we go on to formulate the problem, and we provide a synopsis of the mathematical argument. A general proof strategy for computing the mutual information is also provided.
\item 
Section \ref{mainresult} contains the main theorem -- the asymptotically exact expression for the
mutual information -- and the core of the proof. Since the argument is general enough to capture
expressions deriving from all Renyi entropies, the complete knowledge of the spectrum of reduced states
is also obtained in this fashion. This section reports also the main technical progress when it comes
to dealing with T{\"o}plitz matrices with smooth symbols. 
\item 
Section \ref{sec:temp} is dedicated to a thorough discussion of the temperature dependence of the area law, 
with an emphasis on the the very low temperatures behavior.
\item  We then turn to higher-dimensional free fermionic models on the torus in Section \ref{TorusSec}, 
which we can treat with similar methods.
\item  Finally, we present an outlook in Section \ref{Outlook}, 
comparing the findings with predictions of conformal field theory, and discussing
implications to the study of entanglement spectra \cite{Peschel3,Peschel2,Spectrum,Spectrum2,Spectrum3,Boundary}. 
We also consider the implications for  
noise-driven, open fermionic quantum many-body systems, as have recently been discussed also in the
context of open Majorana wires and notions of noise-driven criticality \cite{EisertProsen,Diehl2,Diehl,Cirac,Timing}.
\end{itemize}

\section{Free fermionic models}\label{Stage}

In this work, we focus on isotropic translationally invariant fermionic models. For now, we assume the most general form for the couplings. Later, we will concentrate on the fermionic variant of the XX spin model (under the Jordan-Wigner transformation), 
which is a particularly important special case. In this section, we 
introduce the class of models discussed here.

\subsection  {Hamiltonians}
We consider free fermionic models on cubic lattices $(\zz_N)^D$ for some even $N\in \nn$ as $N$ becomes large. Bi-sected 
geometries on the torus can be related to the situation of having $D=1$. For simplicity of notation, but without loss of generality, 
we will therefore present this one-dimensional case in the main
part of this work, while discussing the higher-dimensional situation 
in Section \ref{TorusSec}. 
The Hamiltonian takes the general form
\begin{equation}
\label{eqn:Hamiltonian}
	H:= \sum_{i,j\in \zz_N}f_j^\dagger V_{j,k} f_k,
\end{equation}
%
%
%
with $V=V^T \in \rr^{N\times N}$ being a circulant matrix, referred to as the 
{\it Hamiltonian matrix}. We assume the coupling to be 
an specified, finite-ranged interaction of arbitrary interaction length $r\in\nn$.
The couplings are defined by 
a sequence of numbers $(v_k)_{k\in\nn}$ 
with the property that
$v_k=0$ for $k>r$. Then take
\begin{equation}\label{dcoupling}
	d_j:= v_{|j+1|},\, j=-N/2+1,\dots, N/2,	
\end{equation}
having a natural reflection symmetry, $d_{-k \jmod N}= d_{k}$, and consider the Hamiltonian
matrix with entries
\begin{equation}
	V_{i,j}= d_{(i-j)\jmod N}
\end{equation}
for $i,j\in \zz_N$. This situation is depicted in Fig.\ \ref{fig:1d}. 	
The most important special case is constituted by the fermionic variant of the {\it XX model}, which is originally a model for spin-$1/2$ systems. Here, we directly consider its fermionic instance, by virtue of the Jordan Wigner transformation. In this language,
\begin{equation}\label{XX}
	v_1=a,\,
	v_2=b,\,
	v_k= 0
\end{equation}	
for $a,b\in\rr$ and $2<k\in \nn$, so that $V$ is a circulant matrix with $a$ on the main diagonal and $b$ on the first 
off-diagonal.
All these models are integrable and exactly
solvable. Thermal states are defined entirely by the collection of second moments of fermionic operators.
The arbitrary finite range of interactions has in the main
text been chosen for simplicity of notation only, and exponentially decaying interactions can easily be 
accommodated as well. 
Since $V$ is circulant, we find the spectrum $\{\varepsilon_k\}$ of $V$ to \je{satisfy}
\begin{equation}
	\frac{1}{N}\sum_{k\in \zz_N}\varepsilon_k \me^{2 \pi \mi k/N}=d_k .
\end{equation}
For the well-studied fermionic variant of the XX model, see Eq.\ (\ref{XX}), one gets the familiar expression
\begin{equation}
	\varepsilon_k = a +2b \cos \left(\frac{2\pi k}{N}\right).
\end{equation}

\begin{figure}
	\begin{center}
	\includegraphics[width=0.7\linewidth]{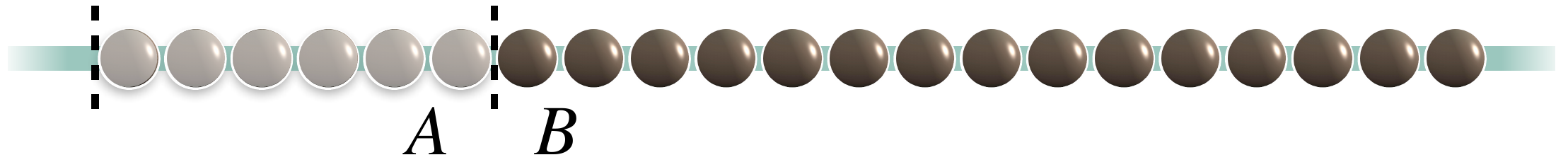}
	\end{center}
	\caption{Geometry of the considered situation for one-dimensional system, consisting of $N$ degrees of freedom, coupled with a local Hamiltonian equipped with periodic boundary conditions.
	The distinguished region embodying sites $\{0,\dots, L-1\}$
	is referred 
	to as $A$, its complement is called $B$.} \label{fig:1d} 				
\end{figure}

\subsection  {Majorana fermions and covariance matrices}
It is convenient to define fermionic {\it covariance matrices} in terms of
{\it Majorana fer\-mions}:
\begin{equation}
   r_i := \frac{f_i + f_i^{\dagger}}{\sqrt{2}},\,
   r_{i+N} := \frac{f_i - f_i^{\dagger}}{\sqrt{2} \mi},\, 
   i=0,\dots, N-1,
\end{equation}
similar to the canonical coordinates for bosonic operators.
These Majorana fermions are
Hermitian, traceless, and form a Clifford algebra.
We are interested in thermal states of the quadratic Hamiltonians, which are particular instances of (quasi)-free or Gaussian fermionic states.
Such states are completely specified by their fermionic covariance matrix
\cite{lagrangefermionic,HalfSpaces,EisertProsen,Terhal}
\hb{$\Gamma\in \rr^{2N\times 2N}$} with entries
\begin{equation}
 \label{eqn:defcovmatrix}
    \Gamma_{i,j}=\mi \tr\left(\rho {[}r_i, r_j{]}\right),
\end{equation}
$i,j=0,\dots, 2N-1$, where the brackets denote the commutator.
\hb{$\Gamma$} is always anti-symmetric, 
\hb{
\begin{equation}
	\Gamma=-\Gamma^T,
\end{equation}}
and satisfies \hb{$-\Gamma^2\leq \id$}. 

\subsection  {Covariance matrices of Gibbs states and their reductions}

Matrix functions of the covariance matrix can be computed exactly.
This is the case because 
any covariance matrix is unitarily equivalent to a direct sum of $2\times 2$ covariance matrices, reflecting a
situation of entirely uncoupled fermionic modes. This is the fermionic analogue of what is often called  the Williamson normal form in the bosonic setting.
As a consequence, one can identify an explicit expression
for the covariance matrix of {\it Gibbs states}
\begin{equation}
	\rho = \frac{\me^{-\beta H}}{\tr(\me^{-\beta H})}
\end{equation}
at inverse temperature $\beta>0$. The covariance matrix $\Gamma\in\rr^{2N \times2N}$
of $\rho$
is given by
\hb{
\begin{equation}
\Gamma = \left( 
\begin{array}{cc}
	  0 & \eta \\
	   -\eta & 0\\
\end{array}
   \right),
\end{equation}
where
\begin{equation}
\eta=f(V)
\end{equation}}
where \je{for all values of $\beta>0$ the smooth function} $f:\rr\rightarrow \rr$ is defined as
\begin{equation}
	f (x):=2\frac{\me^{-\beta x}}{\me^{-\beta x}+1}-1= -\tanh \left(\frac{\beta x}{2}\right).
\end{equation}
\je{The expression $f(V)$ is to be interpreted as a matrix function, that is, $f$ is applied to the spectral values of Hermitian matrices $V$ \cite{Bhatiamatrix}.}
We suppress
the temperature dependence here: Throughout this work, we will be concerned with Gibbs states with
respect to some $\beta$. We allow for arbitrary temperatures and will also later consider the 
asymptotic limits $\beta\rightarrow\infty$ and $\beta \rightarrow 0$. 

Reduced states of Gaussian states are always Gaussian (as can be seen most easily by considering their Grassman representation; 
compare Ref.\ \cite{lagrangefermionic}). The covariance matrix of a reduced state is the appropriate principle sub-matrix of the full covariance matrix. Applied to Gibbs states, we find that the reduced state of a subset $A=\{0,\dots, L-1\}\subset \{0,\dots, N-1\}$ is a Gaussian state 
with covariance matrix
 \begin{equation}
     \Gamma|_A=\left(
\begin{array}{cc}
0 & f(V)|_A\\
-f(V)|_A& 0	
\end{array}
\right),
\end{equation}
where $.|_A$ denotes the principal sub-matrix associated with the degrees of freedom of subsystem $A$. 

\subsection  {T\"oplitz matrices}
T\"oplitz matrices 
\cite{Widom,Widom2,Toeplitzbuch}
may be viewed as principal sub-matrices of infinite circulant matrices. Such matrices will play a crucial role in this work.
Throughout, we will encounter families of {\it T\"oplitz matrices} $T_n\in \rr^{n\times n}$, whose entries are given by 
\begin{equation}
	(T_n)_{i,j} = t_{i-j}
\end{equation}
for some sequence of reals $(t_k)_{k\in \zz}$. 
Colloquially speaking, such T\"oplitz matrices largely resemble circulant
matrices, with the ``upper right and lower left corners'' deviating from a strict circulant matrix. In fact, most
of the theory on T\"oplitz matrices (T\"oplitz determinants) is in one way or the other concerned with the error made when replacing a T\"oplitz
matrix by a circulant matrix of the same dimension. 
Later on we will consider sub-matrices of large finite circulant matrices. 
It is clear that this sequence of T\"oplitz matrices is defined entirely in terms of the sequence of numbers
$(t_k)_{k\in \nn}$. 
This sequence of numbers - and hence the sequence of T\"oplitz matrices - is most conveniently represented in therms of its symbol, defined as the inverse Fourier transform on the torus $\torus:=\left\{x \in \cc :\left|x\right|=1\right\}$ 
of $(t_k)_{k\in \zz}$; i.e. the {\it symbol} $t\in L_{\infty}(\torus)$\footnote{As usual, $L_p\left(\torus\right)$ denotes the set of equivalence classes of $p-$integrable functions up to functions that vanish almost everywhere on $\torus$, with $p=\infty$ denoting the equivalence classes with finite essential supremum.}
is defined as
  \begin{equation}
   t_k = \frac{1}{2 \pi} \int^{2 \pi}_{0} d\theta t(\me^{\mi \theta}) \me^{-\mi k \theta}.
 \end{equation}
The decay of the Fourier coefficients is directly related to the regularity properties of the symbol. 
The summability of the coefficients (i.e., $\sum_{k \in \zz} \left|t_k\right| < \infty$) is sufficient to ensure $t\in L_{\infty}(\torus)$. The relationship between the decaying behaviour of the Fourier coefficients and the regularity of the symbol will be used very frequently. In fact, 
given our assumptions on the interaction parameters, much better regularity results can be derived as we will show and use later on. Mathematically, $t \in L_{\infty}(\torus)$ is equivalent to the requirement that the associated T\"oplitz operator generates a continuous linear operator on $l^2(\nn)$. Furthermore, note that \hb{the spectra} 
can be expressed entirely in terms of the symbol.

The starting point of our analysis is the following observation:
covariance matrices of subsystems of thermal states of translationally 
invariant fermionic models are well approximated by T\"oplitz matrices.
Since the Hamiltonian matrix $V$ in Eq.\ (\ref{eqn:Hamiltonian})
is circulant, the matrix $\eta$ of a Gibbs state $\rho$ is, for any $\beta>0$, circulant 
as well. 
That is to say,
\begin{equation}
	\eta_{i,j} = l_{i-j \jmod N},
\end{equation}	
for some real sequence $(l_k)_{k\in \zz}$, 
suppressing the temperature dependence. Indeed, 
sub-matrices $\eta|_A$  for a region $A=\{0,\dots, L-1\}$ of the matrix $\eta$
(of the large but finite system)
can be well approximated by the  T\"oplitz matrix $\eta_L\in \rr^{L\times L}$ for large $N$ (see Section \ref{AP}).
Th\hb{e} family of T\"oplitz matrices
$\eta_L$ can be expressed in terms of the symbol
\begin{equation}\label{lambda}
	\lambda = f \circ \varepsilon
\end{equation}
where $\varepsilon\in L_\infty(\torus)$ is defined as
\begin{equation}\label{Epsilon}
	\varepsilon(x)= \sum_{k\in\zz} d_k x^k.
\end{equation}
Again, via the inverse Fourier transform on the torus, one recovers
\begin{equation}
	d_k = \frac{1}{2\pi}\int_0^{2\pi}
	d\theta \varepsilon(\me^{\mi\theta}) \me^{-\mi k \theta},
\end{equation}
where $(d_k)_{k\in\zz}$ 
is just the sequence of numbers that govern the coupling in the Hamiltonian matrix in Eq.\ (\ref{eqn:Hamiltonian}).

\subsection  {Entropies of fermionic Gaussian states}
Given that the states we consider are completely described by their covariance matrices, we are able to explicitly compute their Renyi, and 
specifically von-Neumann, entropies explicitly and efficiently. In particular, we consider covariance matrices of the form
\begin{equation}
	\Gamma= \left(
	\begin{array}{cc}
	0 & X\\
	-X & 0
	\end{array}
	\right)
\end{equation}
where $X=X^T\in\rr^{n\times n}$. For fermionic Gaussian states $\rho$ of $n$ modes with
a covariance matrix of this form
one finds
\begin{eqnarray}
	S_\alpha(\rho)&=& \frac{1}{1-\alpha}\log_2
	\tr(\rho^\alpha) = \tr(s_\alpha(X)),
\end{eqnarray}
where the function $s_\alpha:[-1,1]\rightarrow [0,1]$ is defined as
\begin{eqnarray}
	s_\alpha(x) := \frac{1}{1-\alpha}\log_2
	\left(
	\left(
	\frac{1+x}{2}
	\right)^\alpha
	+ \left(
	\frac{1-x}{2}
	\right)^\alpha\right).
\end{eqnarray}
Taking the limit as $\alpha\rightarrow1$, one recovers the von-Neumann entropy,
\begin{eqnarray}\label{Entropy}
	S (\rho)= \tr(s (X))
\end{eqnarray}
with  $s:[-1,1]\rightarrow [0,1]$ being
 \begin{equation}
 	 s(x)= - \frac{1+x}{2} \log_2 \left(\frac{1+x}{2}\right) - \frac{1-x}{2} \log_2 \left(\frac{1-x}{2}\right).
 \end{equation}
  These expressions will be central for our analysis.  \je{With respect to an operational interpretation, the von-Neumann 
  entropic version of the mutual information is by far the most natural quantity in this context, but we keep the generality at this point, partially also because the results found also hold in this 
  setting. With similar techniques, other Renyi divergences can presumably also be treated.}

\section{Statement of the problem}\label{Statement}

\subsection  {Computing mutual informations}
We will now turn to the main object of our study, the {\it mutual information}; a measure of the correlations between two non-overlapping subsystems. 
We will consider a one-dimensional lattice with $N$ lattice sites and let  $A=\{0, \dots ,L-1\}$ 
constitute one part and $B=\{L, \ldots , N-1\}$ its complement. The
quantum mutual information between $A$ and $B$ is defined as
\begin{equation}
	I(A:B)= S(\rho_A)+ S(\rho_B) - S(\rho_{AB}) ,
\end{equation}
where again, $\rho$ is a state on the entire system, $\rho_A$ ($\rho_B$) is obtained by tracing out subsystem $A$ ($B$). The mutual information naturally captures all correlations,
quantum as well as classical, between $A$
and $B$. It is a meaningful measure of correlation for mixed states, including finite temperature thermal states as a special case. For zero temperature, it reduces
to (twice the) {\it entanglement entropy}. Similarly, expressions of the kind
\begin{equation}
	I_\alpha(A:B)= S_\alpha(\rho_A)+ S_\alpha(\rho_B) - S_\alpha(\rho)
\end{equation}
for an $\alpha>0$ different from $1$ are well defined mutual Renyi entropies. 
We will identify novel formulae for the asymptotic behavior
of the mutual information for Gibbs states of isotopic translationally invariant free-fermionic models, and present bounds that allow to study the limit of large $\beta$ analytically.
The formulae given will be exact in the asymptotic limit of large $N$.
 
\subsection  {Structure of the argument for one-dimensional systems}

A number of steps will be necessary in order to arrive at an asymptotically exact
expression of the mutual information. To start with, we need 
to have a handle on how to make the intuition rigorous that we can compute $S_\alpha(\rho_A)$ (and $S_\alpha(\rho_B)$) as if it was the reduced
state of an infinite system. The result has to fit then with the 
expression for $S_\alpha(\rho)$  for the entire system, in a way that only the boundary terms remain. 
\je{This turns out to be a delicate affair, and requires very different tools than the pure state analysis, where the asymptotically exact entanglement entropy is obtained using the Fisher-Hartwig formalism 
\cite{Jin}}. In the case of thermal states, the so-called ``double scaling limit'' makes sense, where $A$ is taken as a constant fraction of the total system size, and the total system is taken to infinity. 
\je{Here, the reduced states of any part necessarily maintain a system size dependence, and one cannot simple compute spectra of reduced states of infinite systems. This is why the
technical tools developed in Lemma \ref{AE} and Lemma \ref{theo:analyticerror} will be necessary. In fact, we show in these lemmas more than what is needed for the main result, in that bounds exponentially
tight in the system size are being provided. Given the prominent status the computation of entanglement Renyi entropy has in the literature, and since it is important to have
rigorous bounds also available for finite system sizes, we expect 
these bounds to be very valuable even outside the precise context of the present work.}


We aim at computing an asymptotically exact approximation of the entropy of a subsystem
\begin{equation}
     S_\alpha(\rho_A) = \tr \left(s_\alpha\left(X|_A\right) \right),
\end{equation}
where $X\in\rr^{N\times N}$ is given by
\begin{equation}
\label{eqn:Xendlich}
    X_{k,l} := x_{k-l}:=\frac{1}{N} \sum^{N-1}_{j=0} 
    \lambda (\me^{ 2\pi \mi j/{N}})
    \me^{- 2\pi \mi j (k-l)/{N}}.
\end{equation}
The proof strategy is as follows:
\begin{itemize}
\item We perform a continuum limit on the full system
\begin{equation}
\label{eqn:Xunendlich}
   X^{(N)}_{k,l}\rightarrow X^{(\infty)}_{k,l} :=  x^{(\infty)}_{k-l}:=\frac{1}{2\pi} \int^{2 \pi}_{0}  \lambda (\me^{\mi \phi}) \me^{-\mi (k-l)\phi} d\phi,
\end{equation}
and show that the entropy of a subsystem can be computed with an exponentially small error in $N$.

\item When computing entropies of sub-matrices $S_\alpha(\rho_A)$, we can in this continuum limit invoke the
theory of T\"oplitz matrices, even though the entries of the T\"oplitz matrices actually depend (very slightly) on the system size. We can therefore
consider families of T\"oplitz matrices with symbol
\begin{equation}
	\lambda = f \circ \varepsilon.
\end{equation}
The same approach is feasible for $S_\alpha(\rho_B)$.

\item We apply trace formulae of T\"oplitz matrices with smooth symbols. This will allow us
to compute an asymptotically exact expression of the entropy of both subsystems.

\item We find an expression for the entropy of the total system $S_\alpha(\rho)$, asymptotically exact in the limit of large $N$.

\item The expression for the mutual information found in this way will contain an infinite sum, reflecting the infinite number of modes in momentum space in the asymptotic limit. As such, this formula can not be evaluated. 
This obstacle we will overcome by a new technique that we introduce, 
making use of the theory of highly oscillatory integrands. This technique is expected to be of significant
interest also outside this context,
when analysing properties of families of T\"oplitz matrices with smooth symbols.
\end{itemize}
The combination of these steps will allow us to provide an asymptotically exact rigorous and computable expression for the mutual information.

\section{An exact asymptotic expression for the mutual information}
In this section, we will state the main results for one-dimensional systems. We will first present the main theorem 
and then continue 
with the proof of the theorem. This will require a number of techniques that we will lay out in later subsection s. The specifically
important case of the XX model, the temperature dependence as well as the situation of free fermions on the torus will be discussed in
later sections. 

In accordance with the analysis of the previous section, we will evaluate the mutual information by analyzing the determinants of the covariance matrices, as the system is taken to the thermodynamics limit. Covariance matrices of free fermionic systems in one dimension are T\"oplitz matrices, so we are free to use the tools available from the theory of T\"oplitz determinants. In particular, it is well known that T\"oplitz determinants behave in a very regular manner as the dimension of the matrices are taken to infinity. In the case when the symbol is continuous, Szeg\"o's strong limit theorem \cite{Szego} gives the precise scaling of the determinant with the dimension of the T\"oplitz matrix. When the symbol has discontinuities, then it is necessary to use the Fisher Hartwig formula in order to obtain accurate asymptotics. As we will be dealing with thermal states, the symbols will always be continuous, even though a discontinuous symbol reflecting the Fermi surface
can be arbitrarily well
approximated for low temperatures. This creates a quite intriguing situation: We can ``approximate'' the situation covered by 
the Fisher-Hartwig theorem arbitrarily well with smooth symbols, and can hence ``interpolate'' between these situations.
As mentioned before, the methods developed here are expected to be useful also outside the context of quantum
many-body systems.

\subsection{Explicit result for one-dimensional systems}\label{mainresult}
We consider 
the mutual information in the large $N$ limit, where
$|A|= L= \theta(N)$ and  $|B|=\theta(N)$ in asymptotic Landau notation,
meaning that both subsystems grow essentially linearly in $N$. Specifically we assume
\begin{equation}
\label{eqn:whatisA}
A= \left\{0,1, \ldots, \lceil qN \rceil-1\right\}
\end{equation}
and
\begin{equation}
\label{eqn:whatisB}
B= \left\{ \lceil qN \rceil,\lceil qN \rceil+1, \ldots, N-1 \right\}.
\end{equation}
The main result for one-dimensional bi-sected systems can be stated as follows. 

\begin{theorem}[Asymptotically exact expression for the mutual information]\label{MainTheorem} 
For any inverse temperature $\beta>0$, and
for any $\alpha\in[0,\infty)$, the mutual information is given by the, in $N$ asymptotically exact, expression
\begin{equation} 
\label{eqn:main}
I_\alpha(A:B)= 
\frac{1}{4 \pi^2}\int^{\pi}_{-\pi} d\phi \int^{\pi}_{-\pi} d\theta 
\frac{s_\alpha(\lambda(\me^{\mi\theta})) - s_\alpha(\lambda(\me^{\mi\phi}))}{\lambda(\me^{\mi\theta}) - \lambda(\me^{\mi\phi})} 
\frac{\lambda'(\phi) - \lambda'( \theta)}{\tan \left(({\phi - \theta})/{2}\right)} +o(1) .
\end{equation}
\end{theorem}

\begin{proof}
The proof of Theorem \ref{MainTheorem} follows from Proposition \ref{Momentum}, which still contains an infinite sum in $k$, and Theorem \ref{Dist}, which take care of the 
infinite sum.
\end{proof}

Here, the temperature dependence is only implicit in $\lambda= f\circ \varepsilon$ (since $f$ depends on $\beta$), while $\lambda'$ denotes the derivative of $\lambda$,
which has to be understood in the following way,
\begin{equation}\label{deriva}
	\lambda'(\theta) := \frac{d}{d\theta} \lambda(\me^{\mi \theta}).
\end{equation}
This is an asymptotically exact and simply evaluated expression. 
The most important instance is the one for the von-Neumann mutual information. 


\begin{proposition}[Mutual information as an infinite sum]\label{Momentum} 
For any inverse temperature $\beta>0$, and
for any $\alpha\in[0,\infty)$, the mutual information is given by the, in $N$ asymptotically exact, expression
\begin{eqnarray}
 	I_\alpha(A:B)&=&\frac{1}{2 \pi^2} \sum^{\infty}_{k=1} \int^{\pi}_{-\pi} d\phi  \int^{\pi}_{-\pi} d\theta
	 \frac{s_\alpha \left(\lambda(\me^{\mi \theta})\right)-s_\alpha \left(\lambda(\me^{\mi 	\phi})\right)}
	 {\lambda(\me^{\mi \theta})-\lambda(\me^{\mi \phi})}  
	 \nonumber\\ 
	 &\times& \sin\left(k(\theta - \phi)\right)\left(\lambda'(\phi)-\lambda'(\theta)\right) +o(1)
	   \label{eqn:lemma3}.
\end{eqnarray}
\end{proposition}
The proof of this statement will require some preparation. For clarity of the main argument, we present some of the technical steps in the appendix.

\subsection  {Approximation statements}\label{AP} 
In this subsection, we collect approximation statements that are being used in the proof
of Theorem \ref{MainTheorem}. We have a more detailed look at the expression
\begin{equation}
    X^{(N)}_{k,l} = x^{(N)}_{k-l}=\frac{1}{N} \sum^{N-1}_{j=0} 
    \lambda (\me^{ 2\pi \mi j/{N}})
    \me^{- 2\pi \mi j (k-l)/{N}}
\end{equation}
as well as at
\begin{equation}
   X_{k,l} =  x_{k-l}=\frac{1}{2\pi} \int^{2 \pi}_{0}  \lambda (\me^{\mi \phi}) \me^{-\mi (k-l)\phi} d\phi,
\end{equation}
where the entries of the finite T\"oplitz matrices are always assumed mod $N$. We know that 
\begin{equation}
	S_\alpha(\rho_A)= \tr(s_\alpha(X^{(N)}|_A)),\,\,\,
	S_\alpha(\rho_B)= \tr(s_\alpha(X^{(N)}|_B)).
\end{equation}
Expression Eq.\ (\ref{eqn:Xunendlich}) 
is the continuum limit of Eq.\ (\ref{eqn:Xendlich}).
Therefore, from a physical perspective, it seems justified to replace $x^{(N)}_{k-l}$ by its infinite-system counterpart $x^{(\infty)}_{k-l}$ in the calculation of the entropy if the system is large enough. In the next lemma we will show that for any fixed temperature, the error made by this replacement
decays exponentially  in $N$.


\begin{lemma}[Approximation of entropies of subsystems]\label{Approx}
\label{AE} For $A,B$ given by Eq.\ (\ref{eqn:whatisA}) and Eq.\ (\ref{eqn:whatisB}) the errors
\begin{equation}
	e_{A,\beta}(N) := |\tr(s_\alpha(X^{(N)}|_A)) - \tr(s_\alpha(X|_A))|
\end{equation}
and
\begin{equation}
	e_{B,\beta}(N):=|\tr(s_\alpha(X^{(N)}|_B)) - \tr(s_\alpha(X|_B))| 
\end{equation}
are exponentially small in $N$ for fixed $\beta>0$.
More precisely, there exist some $N$-inde\-pen\-dent constants $\alpha_{A,\beta}, \alpha_{B,\beta} > 0$, (depending on $\beta$, $q$ and the explicit form of the symbol $\lambda$), such that
\begin{equation}
 \lim_{N \rightarrow \infty} \exp(N \alpha_{X,\beta}) e_{k,\beta}(N) = 0 \textnormal{ for } X \in \left\{A,B\right\}.
\end{equation}
\end{lemma}
This Lemma, along with other auxiliary Lemmas, will be proven in Appendix \ref{ApproxLemma}.

\subsection  {Relating the mutual information to trace functions of T\"oplitz matrices}

Given that we are interested not in the entropy of a subsystem, but in the mutual information, we will need to consider a second order  asymptotic formula 
for the entropy functionals. This is due to the fact that the bulk contribution of the entropy is cancelled out in the expression for the mutual information, and we are left only with the contribution originating
from the boundary. 
We will therefore require a Szeg\"o's strong limit theorem which contains explicit expressions for the second order contributions. In order to formulate the theorem, we will have to introduce some definitions and notation. 
It will be necessary to consider sequences of truncations of T\"oplitz matrices -- in our case principal sub-matrices of 
covariance matrices. 
We will define the following classes of symbols. The {\it Wiener algebra} 
\begin{equation}
   W:= \left\{ \mu \in L_{\infty}\left(\torus \right) : \sum_{n \in \zz} \left| \frac{1}{2 \pi} \int^{2 \pi}_{0}\mu(\me^{\mi \theta}) \me^{-\mi n \theta} d\theta \right| < \infty  \right\},
\end{equation}
the {\it Besov space}
\begin{equation}
    B^{\frac{1}{2}}_2:=\left\{\mu \in L_2\left(\mathbb{T}\right): \sum_{n \in \zz}(\left|n\right|+1)\left| \frac{1}{2 \pi} \int^{2 \pi}_{0} \mu(\me^{\mi \theta}) \me^{-\mi n \theta} d\theta \right|^2<\infty  \right\},
\end{equation}
the {\it Krein algebra}
\begin{equation}
  	W \cap B^{\frac{1}{2}}_2 ,
\end{equation}
and the space of piecewise continuous symbols, denoted by
$PC(\mathbb{T})$ play an important role; see Ref.\ \cite{Toeplitzbuch} for more details. 
Consider a continuous symbol $\lambda: \torus \rightarrow \cc$ and a point $x \notin \bild(\lambda)$. There exists a (unique up to some constant offset of the form $2\pi k$ with $k \in \zz$) continuous argument function 
\begin{equation}
	\arg(\lambda-a): {[}-\pi,\pi{]} \rightarrow \rr; \phi \mapsto \arg\left[\lambda (\me^{\mi \phi} ) - a\right] .
\end{equation}
Independent of the choice of offset, the winding number
\begin{equation}
	\wind(\lambda,a):=\frac{1}{2 \pi}\arg(\lambda-a)(\pi) - \arg(\lambda-a)(-\pi)
\end{equation}
is a well-defined integer.
Before formulating Szeg\"o's limit theorem, we state an important lemma about the spectrum of T\"oplitz operators
(explaining where the spectrum of the truncated matrices $T_n$ eventually concentrates).

\begin{lemma}[Spectrum of T{\"o}plitz operators and truncated T{\"o}plitz matrices \cite{Toeplitzbuch}]
\label{theo:asymptoticspectrum}
Suppose $\lambda \in PC(\mathbb{T})$  and let $T(\lambda)$ be the infinite T\"oplitz matrix associated to the symbol $\lambda$. Assume $U \subseteq \cc$ is an open set and assume that the spectrum $\sigma(T(\lambda))$ is a subset of $U$.
Then there exists an index $n_0 \in \nn$ such that $n \geq n_0$ implies $\sigma(T_n(\lambda)) \subseteq U$. 
If $\lambda$ is continuous the spectrum of $T(\lambda)$ follows entirely from geometric properties of the symbol,
\begin{equation}
\sigma(T(\lambda))= \bild(\lambda) \cup \left\{x \in \cc \left| \wind(\lambda,x) \neq 0\right.\right\} .
\end{equation}
\end{lemma}

Szeg\"o's Theorem holds for symbols in the Krein algebra \cite{Toeplitzbuch}:

\begin{theorem}[A second order Szeg\"o theorem \cite{Widom,Toeplitzbuch}]
\label{th:Szego}
Let $\mu \in W \cap B^{\frac{1}{2}}_2$ be a symbol and let $T_n\in\cc^{n\times n}$ 
be the family of 
associated T\"oplitz matrices. 
Let $\Omega \subseteq \cc$ be an open subset that contains the spectrum of $T(\mu)$.
For an analytic function $g:\Omega\rightarrow \cc$ 
the following trace formula holds, 
\begin{equation}
 \tr \left(g(T_n)\right) = n G_g(\mu) + E_g(\mu) + o(1)
\end{equation}
 with 
\begin{eqnarray}
 G_{g}(\mu) &=&  \frac{1}{2 \pi} \int^{2 \pi}_{0}g(\mu(\me^{\mi \theta})) d\theta ,\\
  E_{g}(\mu)&=& \frac{1}{2 \pi \mi} \int_{\partial \Omega} g(\lambda) \frac{d}{d\lambda} \log E(\mu-\lambda) d \lambda  ,
  \label{ToeplityResidue}\\
  E(\mu) &=& \exp \sum^{\infty}_{k=1}k  \left(\log \mu\right)_k \left(\log \mu\right)_{-k}, 
  \end{eqnarray} where 
  \begin{equation}
  	\left(\log\mu\right)_k :=  \frac{1}{2 \pi} \int^{2 \pi}_{0}\me^{-\mi k \theta} \log\mu(\me^{\mi \theta}) d\theta 
\end{equation}
is the Fourier transform of $\log(\mu)$.
\end{theorem}

\begin{theorem}[Alternative expression for the second order term \cite{Widom2}]
\label{th:Widomformula}
Let $\mu \in W \cap B^{\frac{1}{2}}_2$ be absolutely continuous and let $T_n\in\cc^{n\times n}$ 
be the family of  associated T\"oplitz matrices. Let $\Omega \subseteq \cc$ be an open subset that contains the spectrum of $T(\mu)$ and let $g: \Omega \rightarrow \cc$ be analytic on $\Omega$. Then the second order term $E_g(\mu)$ from Theorem~\ref{th:Szego} can 
be written as
\begin{eqnarray}
 	E_g(\mu)&=&\frac{1}{4 \pi^2} \sum^{\infty}_{k=1} \int^{\pi}_{-\pi} d\phi  \int^{\pi}_{-\pi} d\theta
	 \frac{g \left(\mu(\me^{\mi \theta})\right)-g \left(\mu(\me^{\mi 	\phi})\right)}
	 {\mu(\me^{\mi \theta})-\mu(\me^{\mi \phi})}  
	 \nonumber\\ 
	 &\times& \sin\left(k(\theta - \phi)\right)\left(\mu'(\phi)-\mu'( \theta)\right) 
	   \nonumber.
\end{eqnarray}
\end{theorem}
Again, the derivative is to be read as in Eq.\ (\ref{deriva}).

{\it Proof of Proposition \ref{Momentum}}. 
We are now in the position to prove Proposition \ref{Momentum}.
We can now collect the results from the previous sections. 
By the results from Subsection \ref{AP}, specifically Lemma
\ref{AE},
we know that for any subset of sites whose cardinality 
is much larger than $\beta$, we can work with infinite 
truncated T\"oplitz matrices and their symbols. 
The asymptotic expression for the entropy of the full system  can be obtained directly from the continuum approximation.
Remember that $|A|= \lceil qN \rceil$ 
and subsystem $B$ is of size $|B|=N-|A|$, where $q\in(0,1) $.
Theorem~\ref{th:Szego} states that the block entropies are asymptotically equal to 
\begin{equation}
   S_\alpha(\rho_A) = \tr \left(s_\alpha(X|_A)\right)= q N G_{s_\alpha}(\lambda) + E_{s_\alpha}(\lambda) +o(1)
\end{equation} 
and 
\begin{equation}
   S_\alpha(\rho_B) = \tr \left(s_\alpha(X|_B)\right)= (1-q) N G_{s_\alpha}(\lambda) + E_{s_\alpha}(\lambda)+o(1).
\end{equation} 

We now turn to the computation of the entropy $S_\alpha(\rho)$ of the entire system. This is subtle, and one cannot
employ the same formula for the larger system of $A\cup B$: Lemma \ref{AE} would no longer be valid, and the 
boundary conditions would not be respected.
But we can still find an asymptotically exact expression.
The spectrum of $X$ is given by 
\begin{equation}
	\left\{
		(f\circ \varepsilon)(\me^{2\pi \mi k/N}): k=1,\dots, N
	\right\}.
\end{equation}
Hence,
\begin{equation}
   S_{\hb{\alpha}}(\rho)= \sum_{k \in \zz_N}(s_\alpha\circ f\circ \varepsilon)(\me^{2\pi \mi k/N}).
\end{equation}
Using Lemma~\ref{theo:analyticerror} again, we find that this expression can be replaced by
\begin{equation}
   S(\rho) = N G_s(\lambda) +o(1).
\end{equation}
Combining the terms, the mutual information is asymptotically equal to
\begin{equation}
   I_\alpha(A:B) = 2 E_{s_\alpha}(\lambda) +o(1).
\end{equation}
Since the symbol $\mu:=f \circ \varepsilon$ is analytic (and therefore absolutely continuous) Theorem~\ref{th:Widomformula} can be applied resulting in the desired expression
\begin{eqnarray}\label{eqn:integralforconst}
 	I_\alpha(A:B)&=&\frac{1}{2 \pi^2} \sum^{\infty}_{k=1} \int^{0}_{2\pi} d\phi  \int^{0}_{2\pi} d\theta
	 \frac{s_\alpha \left(\lambda(\me^{\mi \theta})\right)-s_\alpha \left(\lambda(\me^{\mi 	\phi})\right)}
	 {\lambda(\me^{\mi \theta})-\lambda(\me^{\mi \phi})}  
	 \nonumber\\ 
	 &\times& \sin\left(k(\theta - \phi)\right)\left(\lambda'(\me^{\mi \phi})-\lambda'(\me^{\mi \theta})\right) + o(1)
	   \nonumber.
\end{eqnarray}
which completes the proof.\qed

\subsection  {Simplifying Widom's second order expression}

We have almost proven the main theorem, but we still need to eliminate the infinite sum in $k$ in Eq. (\ref{eqn:lemma3}). 
The sum in $k$ is rather awkward; since the integrand is more and more oscillatory for larger and larger 
$k$, it is a priori far from clear 
where one may 
truncate the sum in order to arrive at a reliable result. In this way, the expression cannot be easily computed, not even numerically,
and practically only for low temperatures. We will hence go a technical step further and will prove the validity of Theorem  
\ref{MainTheorem}, the main result for one-dimensional systems. A first step in this direction is the following observation.

\begin{lemma}[Function kernel] For any $n\in\nn$ and any $\phi,\theta\in [0,2\pi)$,
\begin{equation}
   \sum_{1 \leq k \leq n} \sin(k \phi ) =
    \frac{1}{2} \frac{\cos ({\phi}/{2}) - \cos ((n + {1}/{2})\phi)}{ \sin(\phi/2)}=: K_n(\phi).
\end{equation}
\end{lemma}

\begin{proof} 
 We first note that the sum over sine functions encountered 
 in Eq.\ (\ref{eqn:integralforconst}) can be brought into a closed form reminiscent 
 of the Dirichlet kernel. Indeed, if we truncate the sum at $n$, then it can be expressed in closed form as
\begin{eqnarray}
	\sum_{1 \leq k \leq n} \sin(k\phi)&=& \frac{1}{2\mi} \sum_{1 \leq k \leq n}  (\me^{\mi k\phi}-e^{-\mi k\phi})\\
	&=&\frac{1}{2\mi} \sum_{1 \leq k \leq n}  \left((\me^{\mi\phi})^k-(\me^{-\mi\phi})^k\right)\nonumber \\
		&=&\frac{1}{2\mi} 
	\left(
	\frac{1-e^{(n+1)\mi\phi}}{1-e^{\mi\phi}}
	- 
	\frac{1-e^{-(n+1)\mi\phi}}{1-e^{-\mi\phi}}
	\right)\nonumber\\
		&=& \frac{1}{2} \frac{\cos\left({\phi}/{2}\right) - \cos\left((n + {1}/{2})\phi\right)}{ \sin\left(\phi/2\right)},\nonumber
\end{eqnarray}
where the last line follows from elementary trigonometric identities. \qed
\end{proof}

That is to say, the mutual information can be re-expressed in terms of the kernel $K_n$ as
\begin{eqnarray}
\label{eqn:miwithsummedkernel}
 	I_\alpha(A:B)&=&\frac{1}{2 \pi^2} \lim_{n\rightarrow\infty} \int^{\pi}_{-\pi} d\phi  \int^{\pi}_{-\pi} d\theta
	 \frac{s_\alpha \left(\lambda(\me^{\mi \theta})\right)-s_\alpha \left(\lambda(\me^{\mi 	\phi})\right)}
	 {\lambda(\me^{\mi \theta})-\lambda(\me^{\mi \phi})}\\ 
	 &\times& \left(\lambda'(\me^{\mi \phi})-\lambda'(\me^{\mi \theta})\right) K_n\left((\theta - \phi)\right) + o(1)
	   \nonumber.
\end{eqnarray}
We can further simplify the expression by invoking results from the theory of highly oscillatory integrals. Recall the subsequent
fundamental lemma.

\begin{lemma}[Riemann-Lebesgue]
Let $f \in L_1({[} a,b {]},\cc)$, let $g \in L_{\infty}(\rr,\cc)$ be periodic (with period $T$) and assume $\int^{T}_{0}g(x) dx =0$ (the probably most important special instances of such functions $g$ are the trigonometric functions $\sin$ and $\cos$) then 
\begin{equation}
 \lim_{n \rightarrow \infty} \int^{b}_{a}f(t)g(n t)dt = 0 .
\end{equation}
\end{lemma}
\begin{proof}
If $f \in C_C^{\infty}({[} a,b {]},\cc)$ 
(hence is a compactly supported smooth function from $[a,b]$ to $\cc$),
then partial integration gives
\begin{equation}
   	\int^{b}_{a} f(t) g(n t)dt = - \frac{1}{n} \int^{b}_{a} f'(t) \int^{nt}_{a}g(s)ds dt .
\end{equation}
Note that $f'(t)$ is bounded by assumption and that $ |\int^{nt}_{0}g(s)ds | \leq T\left\|g\right\|_{\infty} $ for every $t \in \rr$.
Hence,
\begin{equation}
	\lim_{n \rightarrow \infty} \int^{b}_{a} f(t) g(n t)dt = 0. 
\end{equation}
For general $f \in L_1({[} a,b {]},\cc)$, note that the smooth, compactly supported functions are dense in $L_1({[} a,b {]},\cc)$. 
Hence, for an arbitrary $\varepsilon>0$ there exists $f_2 \in {C_{C}}^{\infty}({[} a,b {]},\cc)$ with 
\begin{equation}
	\left\|f_2 - f\right\|_1 \leq \frac{\varepsilon}{\left\|g\right\|_{\infty}(b-a)}. 
\end{equation}
Therefore,
\begin{equation}
 \lim_{n \rightarrow \infty} \left|\int^{b}_{a} f(t) g(n t)dt\right| \leq \lim_{n \rightarrow \infty} \left(\left|\int^{b}_{a} (f(t) - f_2(t)) g(n t)dt\right| + \left|\int^{b}_{a} f_2(t) g(n t)dt\right|\right) \leq \varepsilon 
\end{equation}
by validity of the statement for compactly supported, smooth functions. Since $\varepsilon>0$ was arbitrary, 
the result also holds true for general $f \in L_1({[} a,b {]},\cc)$.\qed
\end{proof}

This well-known lemma simplifies the understanding of the limiting behaviour of these kernels $K_n$ action on sufficiently regular test functions.

\begin{theorem}[Distributional convergence of the kernels $K_n$]\label{Dist}
Let $g \in L_1(\torus,\cc)$ be locally Lipschitz continuous at $\me^{\mi\theta} \in \torus$ then
\begin{equation}
	\label{eqn:firstdistridentity}
	 \lim_{n \rightarrow \infty} \int^{\pi}_{-\pi} g(\me^{\mi \phi}) K_n(\phi-\theta) d\phi = 
	 \frac{1}{2} \int^{\pi}_{-\pi} \frac{g(\me^{\mi \phi}) - g(\me^{\mi \theta})}{\tan \left( (\phi - \theta)/{2} \right)} d\phi.
\end{equation}
\end{theorem}
\begin{proof}
Without any loss of generality let $\theta=0$ (otherwise take $\phi':= \phi - \theta$). 
Assume $g(1)=0$ first. By Lipschitz continuity of $g$, there exists some $\varepsilon>0$ and some constant $L>0$ such that $\left|g(\me^{\mi \phi})\right| < L \left|\phi\right|$, whenever $\left|\phi\right|<\varepsilon$. By continuity, 
there exists some constant $M>0$ such that 
\begin{equation}
	\left|\frac{\phi}{\tan (\phi/2)}\right| < M
\end{equation}
for $\phi \in {[}-\pi,\pi{]}$. Therefore
\begin{equation}
  \left|\frac{g\left(\me^{\mi \phi} \right)}{\tan \left(\phi/2\right)}\right| \leq \chi_{\left(-\varepsilon,\varepsilon\right)}(\phi) M\cdot L + \chi_{{[}-\pi,\pi{]} \setminus \left(-\varepsilon,\varepsilon\right)}(\phi) \frac{\left|g(\me^{\mi \phi})\right|}{\tan(\varepsilon/2)},
\end{equation}
where $\chi_A$ is the characteristic function of the (measurable) set $A$. 
Define $\tilde{g}(\phi):=g\left(\me^{\mi \phi} \right) / \tan \left(\phi/2\right)$. By the previous estimate $\tilde{g} \in L_1({[}-\pi,\pi{]},\cc)$ and the Riemann-Lebesgue Lemma can be applied. Applying the Riemann-Lebesgue Lemma twice and using elementary angle-sum identities for trigonometric functions yields
\begin{eqnarray}
	\lim_{n \rightarrow \infty} \int^{\pi}_{-\pi} g(\me^{\mi \phi}) \frac{1}{2} \frac{\cos(\phi/2) - 
	\cos\left((n + {1}/{2})\phi\right)}{ \sin(\phi/2)} d\phi & = & \lim_{n \rightarrow \infty}  
	\frac{1}{2} \int^{\pi}_{-\pi} g(\me^{\mi \phi})  
	\frac{1 - \cos(n\phi)}{ \tan(\phi/2)} d\phi \nonumber \\
	& = & \lim_{n \rightarrow \infty}  
	\frac{1}{2} \int^{\pi}_{-\pi} \tilde{g}(\me^{\mi \phi})  
	\left(1 - \cos(n\phi)\right) d\phi \nonumber \\
	& = &   \frac{1}{2} \int^{\pi}_{-\pi} g(\me^{\mi \phi})  \frac{1}{ \tan(\phi/2)} d\phi.
\end{eqnarray}
If $g(1) \neq 0$ write $g(\me^{\mi \phi}) = (g(\me^{\mi \phi}) - g(1)) + g(1)$, apply the former result to the first summand and note that
\begin{equation}
 \int^{\pi}_{-\pi} \frac{\cos(\phi/2) - \cos\left((n + {1}/{2})\phi\right)}{ \sin(\phi/2)} d\phi = 0,
\end{equation}
for every $n \in \nn$ by anti-symmetry of the integrand. \qed

\end{proof}
We need an extension to double integrals, a proof of which follows the same line of reasoning.

\begin{theorem}[Distributional convergence of the kernels $K_n$ for double integrals]
Let $g \in L_1(\torus^2,\cc)$ be Lipschitz continuous in a neighborhood of the diagonal 
\begin{equation}
	D_{\torus^2}:=\left\{(x,y) \in \torus^2 \left| x=y\right.\right\} 
\end{equation}	
	then
\begin{equation}
 \lim_{n \rightarrow \infty} \int^{\pi}_{-\pi}  \int^{\pi}_{-\pi}  g(\me^{\mi \phi},\me^{\mi \theta}) K_n(\phi-\theta) d\phi d\theta = \frac{1}{2}\int^{\pi}_{-\pi} \int^{\pi}_{-\pi} \frac{g(\me^{\mi \phi},\me^{\mi \theta}) - g(\me^{\mi \theta},\me^{\mi \theta})}{\tan \left(\frac{\phi - \theta}{2} \right)} d\phi d \theta.
\end{equation}
\end{theorem}
As a consequence, Eq.\ (\ref{eqn:miwithsummedkernel}) can be rewritten in the following way:
\begin{equation}
 	I_\alpha(A:B)=\frac{1}{4 \pi^2} \int^{\pi}_{-\pi} d\phi  \int^{\pi}_{-\pi} d\theta
	 \frac{s_\alpha \left(\lambda(\me^{\mi \theta})\right)-s_\alpha \left(\lambda(\me^{\mi 	\phi})\right)}
	 {\lambda(\me^{\mi \theta})-\lambda(\me^{\mi \phi})} \frac{\left(\lambda'(\phi)-\lambda'(\theta)\right)}
	 {\tan \left((\phi-\theta)/2\right)} + o(1).
\end{equation}
Let us conclude this subsection  by commenting on the second order term in the trace formula of Theorem \ref{th:Widomformula}.
\je{In unpublished work that we learnt of upon completion of our work}, another related approach aiming at a simplification of Widom's formula has 
been discussed: The result of a \je{masters' thesis} \cite{MasterVassiliev} is 
a derivation of a second order trace formula from asymptotic inverses and appropriate factorisations
 of the symbol. The presented final expression of the second order term does not contain $k$-sum as well and finally only involves the calculation of an appropriate double integral -- a result 
similar to ours. \je{The real-valued integration approach presented here is elementary and readily gives rise to a computable formula in context at hand.}

Indeed, the idea to approach the trace formula starting with the inverse function rather then the usual logarithm 
(the term $E(\mu)$ from Szeg\"o's theorem is just the second order term of the logarithm) has many didactical and computational advantages in our opinion. First of all, this computationally inappropriate and practically incomputable $k-$sum in $E_h(\mu)$ 
is missing right from the beginning -- note that 
for larger and larger $k$, the integrand becomes more and more oscillatory in a fashion that is hard to grasp. 
Furthermore, there exist well-known asymptotic expansions for asymptotic inverses even for higher order terms 
(which are very useful for approximate solutions of systems of linear equation with
 coefficient matrix being a T{\"o}plitz matrix for example). Last but not least the authors of Ref.\
\cite{MasterVassiliev} were able to derive second order formulas for block T{\"o}plitz matrices (which require -- unlike the final result in the one-dimensional setting -- an explicit factorisation of $\lambda \id-\mu$ where $\mu$ is the matrix-valued symbol and $\lambda$ is an arbitrary complex number on the integration contour).

\begin{figure}
{
\begin{minipage}{0.98\linewidth}
  \vspace{0.2cm}
  \begin{minipage}{0.98\linewidth}
  \begin{minipage}[b]{0.48\linewidth}
			\centering
				\mbox{{\bf Critical phase} (with $b=1$, $a=1$)}
	\end{minipage}
	\hspace{0.05\linewidth}
	\begin{minipage}[b]{0.48\linewidth}
			\centering
				\mbox{{\bf Non-critical phase} (with  $b=1$, $a=4$)}
	\end{minipage}
	\vspace{0.01cm}
	\end{minipage}
	\vspace{0.3cm}
  \begin{minipage}[b]{0.98\linewidth}
			\centering
	\end{minipage}
	 \begin{minipage}[b]{0.48\linewidth}
			\centering
				\includegraphics[width=1.25\linewidth]{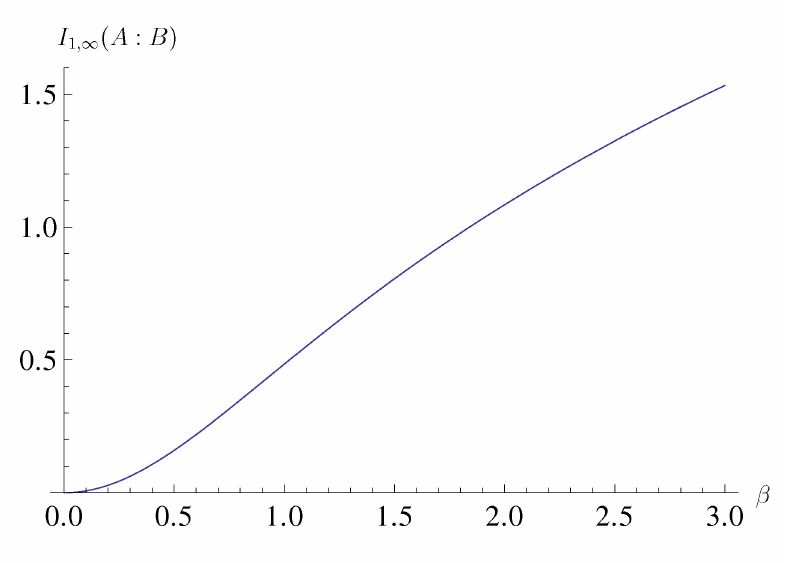}
	\end{minipage}
	\hspace{0.01\linewidth}
	\begin{minipage}[b]{0.48\linewidth}
			\centering
				\includegraphics[width=1.25\linewidth]{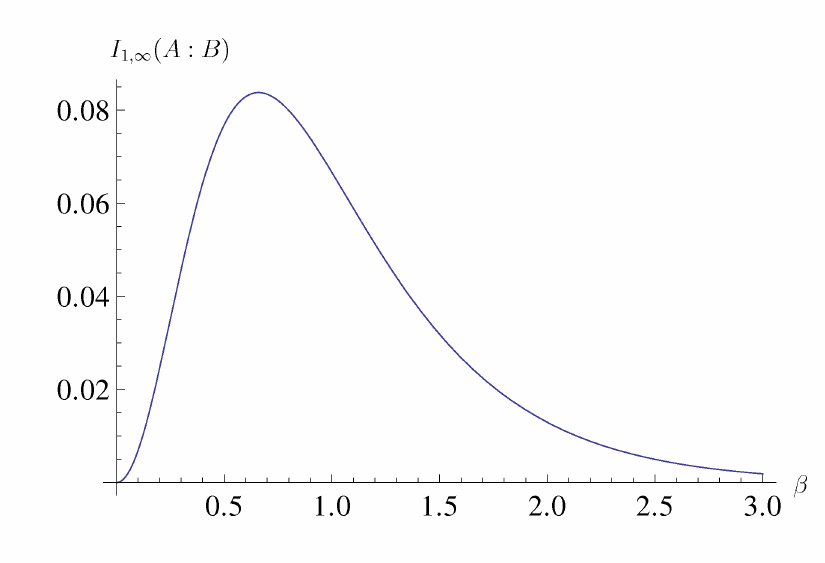}
	\end{minipage}
	\vspace{0.1cm}
  \begin{minipage}[b]{0.98\linewidth}
			\centering
				\mbox{High temperature asymptotics}
	\end{minipage}
	\begin{minipage}[b]{0.98\linewidth}
	\centering
					\mbox{\footnotesize $(\textnormal{where }R_1(\alpha,b):={\alpha b^2}/({2 \log(2)}))$}
	\end{minipage}	
\begin{minipage}[b]{0.48\linewidth}
			\centering
				\includegraphics[width=1.25\linewidth]{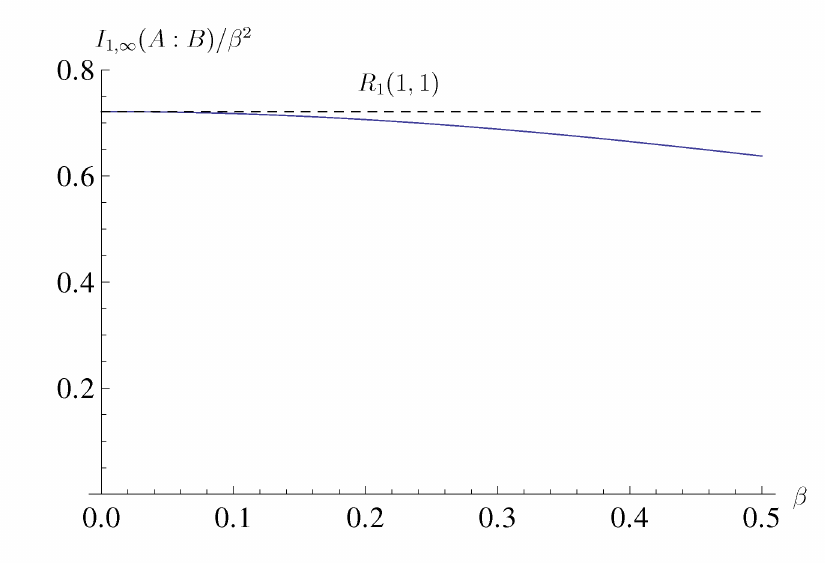}
	\end{minipage}
	\hspace{0.01\linewidth}
	\begin{minipage}[b]{0.48\linewidth}
			\centering
				\includegraphics[width=1.25\linewidth]{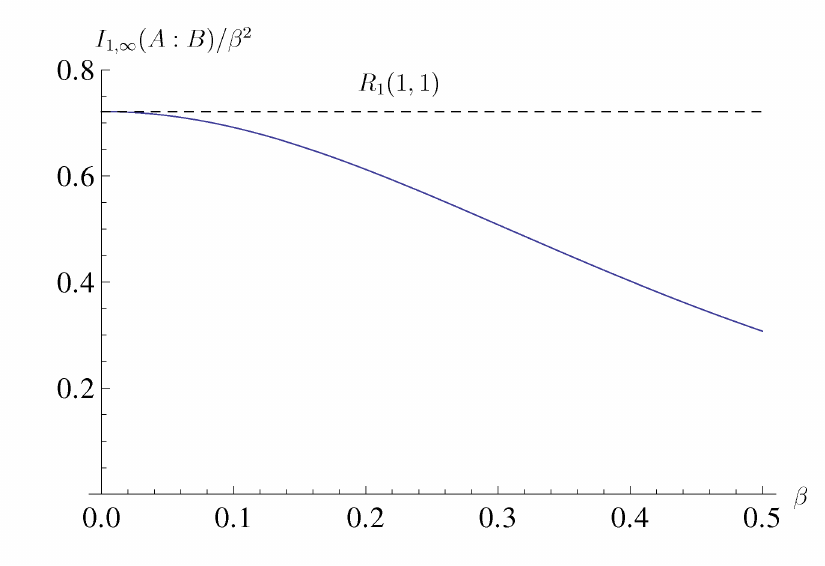}
	\end{minipage}
		\vspace{0.1cm}
  \begin{minipage}[b]{0.98\linewidth}
			\centering
				\mbox{Low temperature asymptotics}
				\mbox{\footnotesize $( \textnormal{where }R_2(\alpha):={\alpha b^2}/({\pi^2 \log(2)}) \min \left\{\frac{27}{\left(\min\left\{a,1\right\}\right)^2},\frac{8}{\left|1-\alpha\right|}\right\}, \hspace{0.1cm} R_3(a,b):=-\left(\left|a\right|-2b\right))$} 
	\end{minipage}
	\begin{minipage}[b]{0.48\linewidth}
			\centering
				\includegraphics[width=1.25\linewidth]{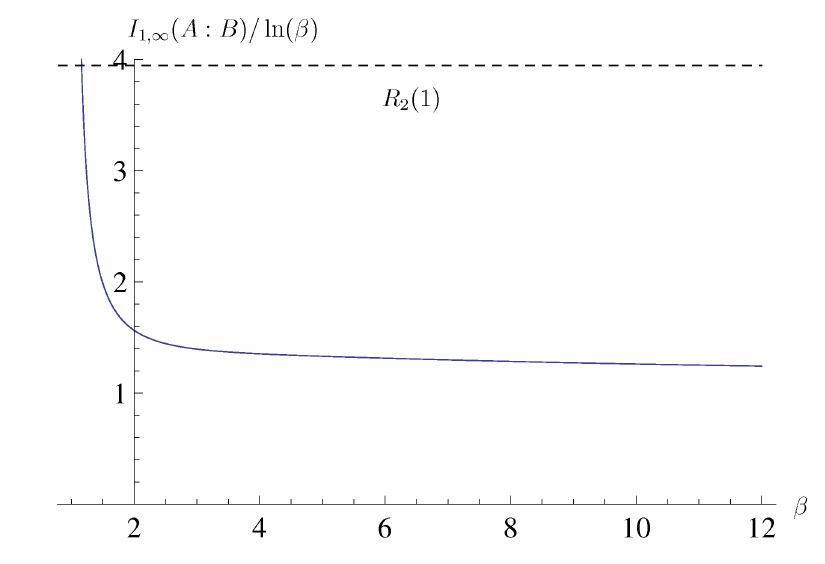}
	\end{minipage}
	\hspace{0.01\linewidth}
	\begin{minipage}[b]{0.48\linewidth}
			\centering
				\includegraphics[width=1.25\linewidth]{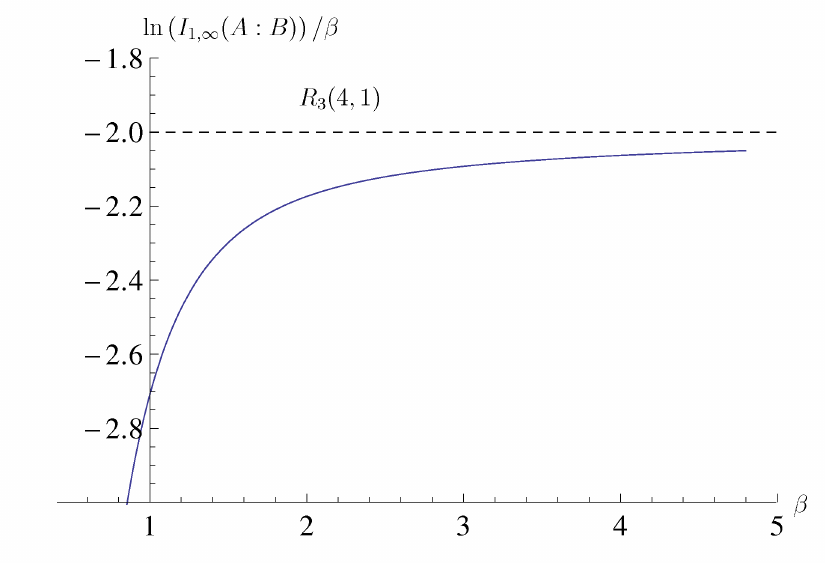}
	\end{minipage}
	\end{minipage}}
	\caption{Temperature dependence of the von-Neumann mutual information for 
	various parameters of the fermionic instance of the XX model.}
	\label{fig:Pic11}
\end{figure}

\section{Temperature dependence}\label{sec:temp}

\subsection  {General remarks}
Our main result, Theorem \ref{mainresult}, gives an asymptotically exact expression for the mutual information 
of two neighbouring blocks of fermions in one dimension. It has been shown that for a given inverse temperature $\beta>0$,
\begin{equation}
   I_{\alpha}(A:B)= I_{\alpha,\infty}(A:B) + o(1)
\end{equation}
in the system size $N$. The result given constitutes an easily computable expression that paves the way for studying 
a number of physically meaningful regimes. At this point, we would like to pause for a moment, however, and would like
to come back to one of the questions posed in the introduction, namely
of the possible asymptotic behaviour of $I_{\alpha,\infty}(A:B)$ for large and small inverse 
temperatures.
As we have seen above, the general bound of Ref.\ \cite{Mutual} following from the extremality of the free energy
suggests that the mutual information should scale like $\beta$
for large inverse temperatures,
\begin{equation}
	 I_{\alpha,\infty}(A:B)= O(\beta).
\end{equation}
One might wonder whether this bound actually 
gives the proper asymptotic scaling on the temperature. Conformal
field theory actually suggests a behaviour which is  logarithmic rather than linearly in the inverse temperature \cite{UniversalKorepin,Area2}. In this section, we corroborate the prediction from conformal field theory by showing that the low temperature asymptotics are given by
\begin{equation}
	 I_{\alpha,\infty}(A:B)= O(\log(\beta)).
\end{equation}
This has a quite remarkable consequence: In order to see features of the criticality of the ground state, one has to go to
extremely low temperatures. This dependence is also convincingly depicted in Fig.\ \ref{fig:Pic11}, where  
a logarithmic scale has been chosen -- otherwise, signatures of ground state features could hardly be detected.
Only at extraordinarily low temperatures, the familiar logarithmic divergence in the 
system size of the sub-block chosen becomes visible for a reasonably sized subsystem.

\subsection  {Analysis of the XX model}

Needless to say, the low temperature asymptotics depends on the choice of the model parameters $a$ and $b$. Specifically,
the scaling of the mutual information reflects signatures of the zero temperature 
quantum phase transition of this model taking place at $\left|a\right|=2b$. Whenever $\left|a\right| < 2b$ the model is critical in its
ground state, i.e., there is no energy gap. It is known that the entanglement entropy then exhibits a logarithmic divergence -- 
signatures of that are also seen in the mutual information at small but non-zero temperature. If the model is gapped, i.e., 
whenever  $\left|a\right| > 2b$, the ground state is the vacuum or the fully occupied state, depending on the sign of $a$. As expected, 
the mutual information converges exponentially quickly to zero in the non-critical phase at a rate proportional to the energy gap. The main insights for this model are summarised in the  following theorem:

\begin{theorem}[Temperature dependence] 
\label{theo:tempdep}
The  ($\alpha-$) mutual information of the thermal state of the  fermionic XX model with parameters $a \in \rr$ and $b>0$ satisfies
\begin{equation}
\label{eq:hightemplimit}
	\lim_{\beta\rightarrow 0}  \frac{I_{\alpha,\infty}(A:B)}{\beta^2}=\frac{\alpha b^2}{2 \log(2)} .
\end{equation}
Whenever $|a|<2b$, then
\begin{eqnarray}
\label{eq:lowtemplimit1}
&& \limsup_{\beta \rightarrow \infty}	\frac{I_{\alpha,\infty}(A:B)}{\log \beta} \\
&\leq & \frac{2 \alpha}{\pi^2 \log(2)\left(\min\left\{1,\alpha\right\}\right)^2}  \min\left\{\frac{27}{2  \left(\min\left\{1,\alpha\right\}\right)},\frac{4}{\left|1-\alpha\right|}\right\} .\nonumber
\end{eqnarray}
Whenever $|a|>2b$, then for any $\kappa < (\left|a\right| - 2 b) \cdot \min\left\{\alpha,1\right\}$
\begin{equation}
\label{eq:lowtemplimitferro}
\lim_{\beta \rightarrow \infty} \exp(\kappa \beta) I_{\alpha,\infty}(A:B) = 0.
\end{equation} 
\end{theorem}

%
%

The proof of this statement is rather involved and requires a number of techniques developed in lemmas:
Hence, for better readability of the main text, it will be presented in
the appendix in Subsection \ref{tempproof}.


\section{Free fermonic models on the torus}\label{TorusSec}

The results established above readily apply to the situation of a bi-sected, higher-dimen\-sional fermionic lattice system on the torus.
A quite similar strategy has already been exploited, e.g., in Ref.\ \cite{HalfSpaces}. 
Using appropriate discrete Fourier transforms, one can 
disentangle the constituents from each other with respect to 
all but one dimensions. The exception is the dimension in which the two regions
 labeled $A$ and $B$ are singled out, see Fig.\ \ref{fig:torus}. In this way, one 
 arrives at a collection of 
 suitably modulated and altered one-dimensional problems, to which the above statements apply. In this section, we highlight the 
results obtained in this manner. An important application of this is the computation of the mutual information in 
{\it higher-dimensional tight binding models}.

\subsection  {Geometry of the problem}

Let us for simplicity consider in $D$ dimensions 
the geometry of slabs $L= (\zz_M)^{D-1}\times \zz_N$, for suitable $N$ and $M$. 
So along one dimension, we have as before $N$ sites, whereas the system size with respect to the
other dimensions is $M$, see Fig.\ \ref{fig:torus}. The index set of all sites can be taken to be
\begin{eqnarray}
	I &=& I'\times J,\\
	I'&=& \{0,\dots, M-1\}^{\times D-1},\\
	J&=&\{0,\dots,N-1\}.
\end{eqnarray}
Vectors of indices $i\in I\subset \zz^D$ 
will be regarded as modulo $N$ and $M$, respectively. Let
$w:\zz^{D}\rightarrow L$ be defined as
\begin{equation}
	w(i) = w(i_1,\dots, i_D) = (i_1 {\rm mod} M,\dots, i_{D-1} {\rm mod} M, i_D{\rm mod} N).
\end{equation}
This function simply projects arbitrary indices from $\zz^D$ onto the lattice $L$. We again allow for arbitrary
finite-ranged interactions (the generalization to exponentially decaying interactions is straightforward but omitted). That is to say, 
for $i,j\in I$, the Hamiltonian tensor takes the form
\begin{equation}
   V^{{i}}_{{j}}=d_{w(i-j)}. 
\end{equation}
Then the Hamiltonian can be written in the following way
\begin{equation}
   H := \sum_{i,j \in I} V^{i}_{j} f_{i}^{\dagger} f_{j} = \sum_{a,b \in I'; c,d \in J} V^{(a,c)}_{(b,d)} f_{(a,c)}^{\dagger} f_{(b,d)}.
\end{equation}
We have singled our the special spatial dimension for which we consider the bi-partite cut. We take the parts $A$
and $B$ to be
\begin{eqnarray}
A&=&I'\times  \left\{0,1, \ldots, \lceil qN \rceil-1\right\},\\
B&=& I'\times \left\{ \lceil qN \rceil,\lceil qN \rceil+1, \ldots, N-1 \right\},
\end{eqnarray}
with $q$ as before. We will also see that the previously found results still apply.

\begin{figure}
	\begin{center}
	\includegraphics[width=0.7\linewidth]{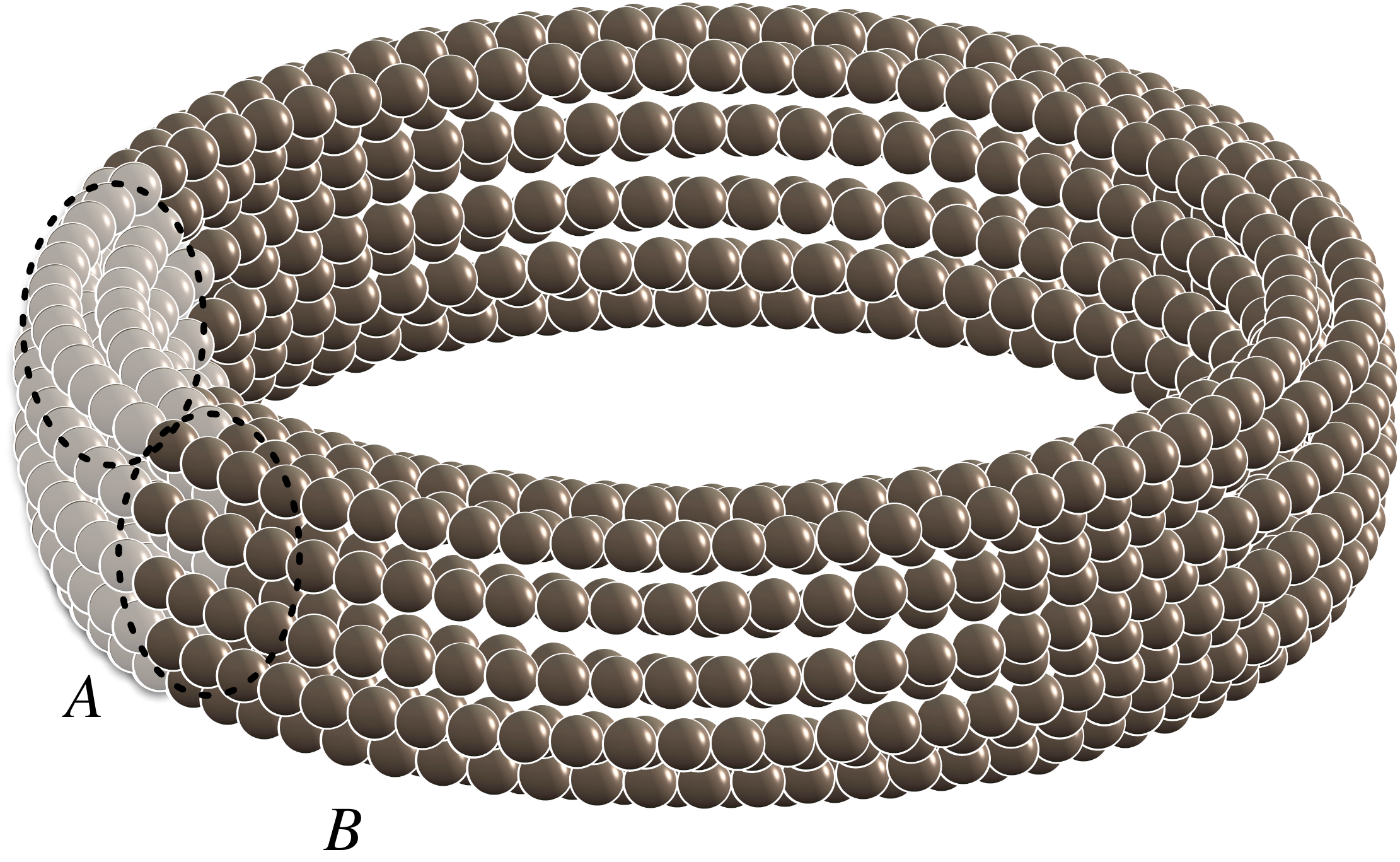}
	\end{center}
	\caption{The geometry of the free fermionic model on the torus. Along a singled out dimension, the system consists of $N$ sites, while
	the other dimensions with the topology of a torus embody $M$ sites. The distinguished region is again referred to as $A$,
	its complement is $B$.} \label{fig:torus} 				
\end{figure}

\subsection  {Scaling of the mutual information on the torus}
We can now consider Fourier transforms with respect to the index set $I'$, while leaving the index set $J$ invariant.
For the discrete Fourier transform, we choose
\begin{equation}
U_{k,l}:=\frac{1}{\sqrt{M^{D-1}}}\exp{\left(\frac{2 \pi \mi}{M^{D-1}} k \cdot  l \right)} \textnormal{ for } k,l \in I'.
\end{equation}
It can then be shown that
\begin{eqnarray}\label{Modu}
	{V'}^{(k,k_0)}_{({l},l_0)}& :=& \left(\left(U \otimes \id \right) \left(V^{{i}}_{{j}}\right)_{{i},{j}} \left(U^{\dagger} \otimes 
	\id \right)\right)_{({k},k_0),({l},l_0)}\nonumber\\
	&=&\underbrace{\left(\sum_{{t} \in I'}\exp{\left(\frac{2 \pi \mi}{N'}{t} \cdot {l}\right)} d_{\left({t},k_0-l_0\right)}\right)}_{:=	
	\tilde{d}^{({l})}_{k_0,l_0}} \delta_{{k},{l}}.
\end{eqnarray}
That is to say, the reduced entropies of $A$ and $B$ as well as the global entropy can be 
computed as if one had uncoupled one-dimensional systems, only that the coupling is modulated in a
way described by the new couplings in Eq.\ (\ref{Modu}). Denote for $k\in I'$ the resulting expression of
the symbol computed from $\tilde d^{(k)}$
\begin{equation}
	\lambda_k = f\circ \varepsilon_k,
\end{equation}	 
where $\varepsilon_k$ is the expression as defined in Eq.\ (\ref{Epsilon}), with $d$ replaced by $\tilde d^{(k)}$ for each $k$. 
In this way,
we also arrive at the asymptotically exact expression for the mutual information in the limit of large $N$, recovering the logarithmic
divergence of the mutual information.

\begin{theorem}[Mutual information on the torus]\label{Torus} 
For any inverse temperature $\beta>0$, any $M\in \mathbb{N}$, and
for any $\alpha\in[0,\infty)$, the mutual information is given by the, in $N$ asymptotically exact, expression
\begin{eqnarray} 
\label{eqn:torus}
I_\alpha(A:B) &=& 
\frac{1}{4 \pi^2}
\sum_{k\in I'}
\int^{\pi}_{-\pi} d\phi \int^{\pi}_{-\pi} d\theta \nonumber\\
&\times &\frac{s_\alpha(\lambda_k(\me^{\mi\theta})) - s_\alpha(\lambda_k(\me^{\mi\phi}))}{\lambda_k(\me^{\mi\theta}) - \lambda_k(\me^{\mi\phi})} 
\frac{\lambda_k'(\phi) - \lambda_k'( \theta)}{\tan \left(({\theta - \phi})/{2}\right)} +o(1).
\end{eqnarray}
\end{theorem}
In particular, in this way, we find for tight binding models on the torus in any dimension $D$
that
\begin{equation}
	I(A:B)= O(M^{D-1}\log(\beta))
\end{equation}
for large inverse temperatures $\beta$: The logarithmic divergence in the inverse temperature remains, while an additional term
relating to an area law in $D$ dimensions emerges.

\section{Outlook}\label{Outlook}

\subsection  {Conformal field theory and entanglement spectra}

In this section, we compare our results with predictions from {\it conformal field theory}. Indeed, 
our findings can be seen as a confirmation of the predictions resulting from the conformal transformation relating finite systems at zero temperature to infinite systems at finite temperature. In our result, we provide system specific qualifiers of the asymptotic limits, and our main theorem constitutes a fully
rigorous result applicable to a large class of models. 
Also, higher-dimensional systems on tori can be captured with the methods
presented here. Still, it is interesting to see that the behaviour of the mutual information, scaling asymptotically 
as the logarithm in the inverse temperature, can also immediately be suggested from the analysis of conformal field
theories in $1+1$ dimensions. This is an observation similar to the one 
made for ground state properties of the XX model in the ground state \cite{Jin,FH}:
Here, the rigorous expression given relates to and confirms the formula \cite{Wilczek,Area2}
\begin{equation}
	S_A = \frac{c}{3}\log \frac{l}{a}+ c_1,
\end{equation}
in a connection already conjectured in Ref.\ \cite{Latorre}.
Here, $c_1>0$ is a constant and $a$ is the lattice spacing, while $c$ is the conformal charge. 
The results on entanglement entropies in the XX models were proven using Fisher-Hartwig type methods
for T\"oplitz determinants \cite{Jin,FH}.

The main idea in the conformal analysis is that there exists a connection between correlation properties of  finite systems at zero temperature and infinite systems at finite temperature. Therefore, assume that at zero temperature, the
entanglement entropy scales as
\begin{equation}
\label{conformalEntropy1}
	S_A = f( l)
\end{equation}
for some function $f:\rr^+\rightarrow\rr^+$. Then, at large inverse temperatures $\beta>0$, we consider  the conformal mapping
$z\mapsto \me^{2\pi z/(v\beta)}$ for $v>0$, to arrive at the expression \cite{UniversalKorepin}
\begin{equation}
	S_A = f\left ( \frac{v\beta}{\pi}\sinh\left(
	\frac{\pi l}{v}
	\right)\right).
\end{equation}
Combined with Eq. (\ref{conformalEntropy1}), one hence gets for $\beta>0$ \cite{UniversalKorepin,Area2}
\begin{equation}
	S_A = \frac{c}{6}\log \left(
	\frac{\beta}{\pi a}\sinh\left(
	\frac{2\pi l}{\beta}
	\right)
	\right)+c_3,
\end{equation}
for $c>3>0$.
The above statement becomes meaningful in the limit when $l\gg \beta$.
Asymptotically,  $S_A$
is well approximated by $\pi c l/(3\beta) + c_3$ \cite{Area2}, but for our purposes, we also need a first  order correction,
namely
\begin{equation}
	S_A \approx \frac{c}{6}\left(\log \left(
	\frac{\beta}{\pi a} 
	\right)
	-\log(2)+ \frac{2\pi l}{\beta}
	\right)+c_3.
\end{equation}
In this way, one recovers the expression for the mutual information
\begin{equation}
	I(A:B)\approx \frac{c}{3}\log \left(
	\frac{\beta}{\pi a}
	\right),
\end{equation}
again exhibiting a logarithmic dependence in $\beta$, as $\beta\rightarrow \infty$. The analysis of the finite temperature case axing arguments from conformal field theory has recently been carried out in Refs.\ \cite{Swingle2,LeHur2}. 
As mentioned earlier, the criticality of the ground state is not felt by the mutual information, unless one is at extremely
low temperatures. In turn, the logarithmic scaling in the size of a subregion is turned into a logarithmic dependence on the inverse
temperature.

It is important to note that because the formalism developed here allows for the computation of all Renyi entropies, the tools developed in this paper also allow for the study spectra of reductions; i.e.  {\it entanglement spectra} \cite{Peschel3,Peschel2,Spectrum,Spectrum2,Spectrum3,Boundary}, 
compare also Refs.\ \cite{Area,Area2}. Such entanglement spectra have turned out to provide a powerful tool 
when characterizing {\it topological order}, relating bulk to boundary theories, and discussing the possibility of
approximating states with suitable {\it tensor network states} such as {\it matrix product states and operators} and higher-dimensional
analogues.

\subsection  {Remarks on open free fermionic many-body systems}

Let us finally briefly mention that the formalism developed here in principle also allows for the study of {\it open quantum systems},
undergoing dissipation and quantum noise. An extensive discussion of this topic is beyond the scope of the present article.
In the mindset developed here, one can consider {\it Liouvillians} capturing fermionic open quantum systems
of the form
\begin{equation}	
	{\cal L}(\rho) = \mi[H,\rho] + \sum_{j\in\zz_N} \left(
	L_j \rho L_j^\dagger - \frac{1}{2}\{
	L_j^\dagger L_j,\rho
	\}
	\right)
\end{equation}
with, as before,
\begin{equation}\label{eqn:Hamiltonian2}
	H= \sum_{i,j\in \zz_N}f_j^\dagger V_{j,k} f_k,
\end{equation}
and the $L_j$ are operators linear in the fermionic operators, supported on finitely many sites only, which are all the same,
except that they are all translates of each other. The Lindblad operators $L_j$ hence act locally in the same way as the Hamiltonian terms
act locally. The equations of motions then become
\begin{equation}
	\frac{d}{dt}\rho(t) = {\cal L}(\rho).
\end{equation}
A state $\rho$ is called a steady state if ${\cal L}(\rho)=0$. For fully translationally invariant systems, one can again define a 
symbol \cite{EisertProsen,Diehl2}, 
and the methods developed here are applicable. Details of such an approach
will be pursued elsewhere. Such an idea seems particularly timely in the light of the observation that 
open system dynamics should not only be viewed as a source of quantum noise added to the system. But that 
open system dynamics and dissipation can also be beneficial: Indeed, dissipative dynamics allows for
dissipative instances of quantum information processing \cite{Cirac,Diehl,Timing}, exhibiting an interesting way of 
protection. Also, steady states can readily exhibit a number of exciting properties: they can be {\it entangled}, show
phenomena of {\it noise-driven criticality}, or even exhibit {\it topological order} \cite{Diehl,Diehl2,EisertProsen}. The latter
is specifically true for free fermionic models, an arena for which the machinery developed here should be fruitful.

\subsection  {Summary} In this work, we have introduced a formalism allowing us to prove
the validity of rigorous expressions for the Renyi mutual information of Gibbs states of translationally invariant quasi-free fermionic models. 
The expressions obtained are asymptotically exact; the bounds given are also exponentially tighter than those
derived from the extremality of the Gibbs state with respect to the free energy. 
Special emphasis has been put onto the technical development of novel methods of dealing with T{\"o}plitz matrices, in particular a 
new and useful instance of a second order expression for smooth symbols. These tools, as well as the approximation results introduced here, 
are expected to be widely applicable and highly useful in various contexts. It is the hope that this work
inspires further entanglement-related studies of mixed fermionic quantum states, arising in the context of 
describing both closed and open quantum many-body systems.

\section*{Acknowledgements}
This work has been supported by the EU (QESSENCE, RAQUEL, COST, SIQS, AQuS), 
the FQXi, the EURYI award scheme, the BMBF (QuOReP), and the ERC (TAQ). 
We acknowledge discussions with M.\ Cramer on entanglement properties of free fermionic systems, with 
B.\ Silbermann on second order trace formulae, \je{and with A.\ Winter on Renyi entropic mutual informations}.

\section{Appendix}

\subsection  {Simple and useful upper and lower bounds to the mutual information}

In this subsection, we will derive upper and lower bounds to the quantum mutual information, evaluated for
Gaussian fermionic states. These bounds are generally useful when bounding mutual information
expressions.
Moreover, they are additive, so share this important feature with the actual mutual information. 
We expect these bounds to be useful also in other contexts, different from the study pursued here. What is more,
it should be clear that one can immediately also formulate a bosonic variant of the bounds presented here, using the 
same strategy of proof.
We consider covariance matrices $\Gamma\in\rr^{2n\times 2n}$
of the form
\begin{equation}
	\Gamma=
	\left(
	\begin{array}{cc}
	0 & X\\
	-X& 0
	\end{array}
	\right),
\end{equation}
with $X=X^T$. We write
\begin{equation}
	X=\left(
\begin{array}{cc}
	X_A & X_{AB}\\
	X_{BA} & X_B
	\end{array}\right)
\end{equation}
and denote with  $P$  the pinching \cite{Bhatiamatrix}
\begin{equation}
	P=\left(
	\begin{array}{cc}
	X_A & 0\\
	0 & X_B
	\end{array}
	\right) = X_A\oplus X_B
\end{equation}
of $X$. This is the covariance matrix of the tensor product of the two reduced states, as a moment of thought reveals.
In these terms, the quantum mutual information can hence be written as
\begin{equation}
	I(A:B)= {\rm tr}(s(P ))-{\rm tr}(s(X )),
\end{equation}	
so as a difference of two trace functions. Let us now consider a quadratic 
function $l:[-1,1]\rightarrow [0,1]$
defined as
\begin{equation}
	l(x)=(1-x^2)/2.
\end{equation}
We will see that the corresponding difference of trace functions is a convenient, computable lower bound
of the mutual information with reasonable properties. We hence relate the mutual information to {\it purities} 
only (or, for that matter, to $2$-Renyi mutual informations)\footnote{It may be important to state 
that since these quantities only referring to local and global purities 
are easier to measure, these bounds are expected to be useful even in an experimental context.}.

\begin{lemma}[Lower bound to the mutual information]
For any real symmetric $X=X^T$,
\begin{equation} 
{\rm tr}(s( P))  - {\rm tr}(s( X))    \geq    {\rm tr}(l( P )) -{\rm tr}(l( X )).
\end{equation}
Moreover, the bound is additive, in that 
\begin{equation}
	{\rm tr}(l( P )) = {\rm tr}(l( X_A ))+ {\rm tr}(l( X_B )).
\end{equation}
\end{lemma}

\begin{proof} To start with, the additivity immediately follows from the definition. We have
\begin{eqnarray}
	{\rm tr}(l( P )) &=& \frac{1}{2}{\rm tr}(\id_{2n}-(X_A\oplus X_B)^2) \nonumber\\
	&=& \frac{1}{2}
	 {\rm tr}(\id_{2n}-(X_A^2\oplus X_B^2))\nonumber\\
	&=&
	{\rm tr}(l( X_A ))+ {\rm tr}(l( X_B )).
\end{eqnarray}
This lower bound as such can be shown using a result of the theory of convex trace functions called Peierls' inequality
\cite{Bhatiamatrix} \hb{(see also Lemma~\ref{lem:Peierl})}. The function
$s-l$ is a concave function (it is not required that it is operator concave). What is more, we can without loss of generality
assume that $X_A$ and $X_B$ are both diagonal. If they are not, there exist real orthogonal matrices $O_A,O_B$
such that $(O_A\oplus O_B)P(O_A\oplus O_B)^T$ is diagonal, without altering the value of the mutual information or
the bound. Using the concavity of $g:= s-l$, we find
\begin{equation}
	\sum_j  g(X_{j,j}) \geq {\rm tr}(g(X)).
\end{equation}		
But since $X_A$ and $X_B$ are diagonal, we have
\begin{equation}
	{\rm tr}(g (P ) ) \geq {\rm tr}(g(X)) =  {\rm tr}(s( X))  -  {\rm tr}(l( X ))
\end{equation}	
which is what we intended to show.\qed
\end{proof}
Note the similarity of this bound for the mutual information with the 
bound presented in Ref.\ \cite{MMW} for the von-Neumann entropy. It is 
quite remarkable that this bound can even be used for differences 
of entropies, not only of entropies as such. 

We now turn to upper bounds. Let us define the function $u:[-1,1]\rightarrow [0,1]$
as
\begin{equation}
	u(x)=\frac{1}{\log(2)}(1-x^2)^{1/2}.
\end{equation}
Again, one can prove that the upper bound derived from this matrix function, ${\rm tr}(u( P))  - {\rm tr}(u( X))$ 
shares the above additivity property. Again, one can prove in the same way as above that one encounters
an upper bound. Of course, the function $u-s$ is no longer concave, but convex, so $l-s$ is concave. 
This bound is not a quadratic bound any more.
Using again Peierls' inequality, one arrives at the following result.

\begin{lemma}[Upper bound to the mutual information]\label{LemInt1}
For any real symmetric $X=X^T$,
\begin{equation} 
{\rm tr}(s( P))  - {\rm tr}(s( X))    \leq    {\rm tr}(u( P )) -{\rm tr}(u( X )).
\end{equation}
The bound is additive,
\begin{equation}
	{\rm tr}(l( P )) = {\rm tr}(u( X_A ))+ {\rm tr}(u( X_B )).
\end{equation}
\end{lemma}
This upper bound on the mutual information is a very useful bound in its own right.

\subsection  {Derivatives of entropy functions}

For convenience, we present the derivatives of the entropy functions used here.

\begin{lemma}[Derivatives of the entropy functions] The derivatives of the entropy functions are given by
\begin{eqnarray}
\label{eqn:firstdersalpha}
s_{\alpha}'(x) &=& \frac{\alpha}{(1-\alpha) \log(2)} \frac{\left(1+x\right)^{\alpha - 1} - \left(1-x\right)^{\alpha - 1}}{\left(1+x\right)^{\alpha} + \left(1-x\right)^{\alpha}},\\
s_{\alpha}''(x) &=& - \frac{\alpha}{\log(2)} \left( \frac{\left(1+x\right)^{\alpha-2} + \left(1-x\right)^{\alpha-2}}{\left(1+x\right)^{\alpha} + \left(1-x\right)^{\alpha}}   \right. \nonumber \\
 &+&\left.\frac{\alpha}{1-\alpha}\left(\frac{\left(1+x\right)^{\alpha-1} - \left(1-x\right)^{\alpha-1}}{\left(1+x\right)^{\alpha} + \left(1-x\right)^{\alpha}}\right)^2 \right), \label{eqn:seconddersalpha}\\
\label{eqn:firstders}
s'(x) &=&\frac{1}{2} \log_2 \left(\frac{1-x}{1+x} \right),\\
\label{eqn:secondders}
s''(x) &=& - \frac{1}{\log(2)} \left(\frac{1}{1+x} + \frac{1}{1-x} \right).
\end{eqnarray}
\end{lemma}

\subsection{Proof of approximation theorems used in the main text}\label{ApproxLemma}

In this subsection, we prove Lemma \ref{Approx} along with further auxiliary statements. This will require considerable effort, but will lead an approximation result
that is expected to be useful also in other contexts than the specific one considered here. Essentially, we present exponentially tight bounds on matrix entries of
sub-matrices of covariance matrices of large translationally invariant fermionic systems. As such, they can also be used when computing entanglement entropies at
zero temperature.

{\it Proof of Lemma \ref{Approx}:} The proof for the set $B$ is completely analogous to the proof for the set $A$, so we will consider the latter
subset $A$ only.
We will proceed in four steps:
\begin{enumerate}
\item $|x^{(N)}_k - x_{k}|$ is the well-known error of approximating an integral by a sum using the trapezoidal rule.
\item We write the square norm as an appropriate sum of $|x^{(N)}_k - x_{k}|$.
\item We bound the latter by the square norm of the difference of both matrices.
\item We use regularity properties of $s_{\alpha}$ to bound the total error by the error in the eigenvalues originating from the replacement of $X^{(N)}|_A$ by $ X|_A$. 
\end{enumerate}

The function $\phi \mapsto \lambda(\me^{\mi \phi})$ is real-analytic and $2\pi-$periodic with values in some compact subinterval of $\left(-1,1\right)$. By the first part of Lemma~\ref{theo:analyticerror} there exist some constants $\kappa>0$ and $M>0$ such that
\begin{equation}
  \lambda \left(\me^{\mi \phi} \right) \leq M \textnormal{ for all } \phi \in D ,
\end{equation}
where
\begin{equation}
	D:=\left\{x \in \cc : \left|\Im(x)\right| \leq \kappa \right\}.
\end{equation}
This implies
\begin{equation}
 \left|\lambda \left(\me^{\mi \phi} \right) \me^{\mi k \phi}\right| \leq M \me^{\left|k\right|\kappa}
\end{equation}
for all $k \in \zz$. Therefore, by Lemma~\ref{theo:analyticerror} we get
\begin{equation}
	\label{eqn:decayofFouriercoeff}
	\left|x^{(N)}_k - x_{k}\right| \leq \frac{2 M \me^{\left|k\right|a}}{\me^{a N} - 1} \leq \frac{2 M \me^{L a}}{\me^{a N} - 1} = \frac{2 M \me^{\left\lceil qN \right\rceil a}}{\me^{a N} - 1}.
\end{equation}
For the third step note that
\begin{eqnarray}
\label{eqn:squaresum}
 	\left\|X^{(N)} |_A -X|_A \right\|_2^2 
 	=  \sum^{(L-1)}_{j=-(L-1)}(L-\left|j\right|)
 	\left| x^{(N)}_{j}-x_{j}\right|^2,
 \end{eqnarray}
which follows directly from counting the number of entries in both T\"oplitz matrices.\\
For the second step let $\mu^{(A)}_r$ denote the $r$-th eigenvalue of $X_A$  and let $\mu^{(N,A)}_r$ denote the $r$-th eigenvalue of $X^{(N)}_A$. Possibly relabeling these eigenvalues, the Hoffman-Wielandt theorem 
(Lemma~\ref{theo:HoffmanWielandt}) yields
\begin{equation}
\label{eqn:HoffmanWielandtforEigenvalues}
    \sum^{L-1}_{r=0}\left|\mu^{(N,A)}_r - \mu^{(A)}_r\right|^2 \leq \left\|X^{(N)} |_A -X|_A \right\|_2^2.
\end{equation}
Finally, note that the Renyi entropy functions $s_{\alpha}$ are H{\"o}lder continuous (see Lemma~\ref{theo:Hoeldercone}) with some appropriate H{\"o}lder exponent $0 < \gamma \leq 1$. In other words, there exists some constant $C>0$ such that
\begin{equation}
  \left|s_{\alpha}(x) -  s_{\alpha}(y)\right| \leq C \left|x-y\right|^{\gamma}.
\end{equation}
This gives the estimate
\begin{eqnarray}
e_{A,\beta}(N) & \leq &  \sum^{L-1}_{r=0}\left|s_{\alpha}(\mu^{(N,L)}_r) - s_{\alpha}(\mu^{(L)}_r)\right| \label{eqn:Hoeldercontestentropy}\\
              &\leq& C \left(\sum^{L-1}_{r=0}\left|\mu^{(N,L)}_r - \mu^{(L)}_r\right|^{\gamma} \right) \nonumber \\
              &\leq& C \left(\sum^{L-1}_{r=0}\left|\mu^{(N,L)}_r - \mu^{(L)}_r\right|^{2} \right)^{{\gamma}/{2}} L^{1-{\gamma}/{2}} ,\nonumber
\end{eqnarray}
where we have used H\"older's inequality in the third line. Combining Eq.\ (\ref{eqn:Hoeldercontestentropy}), 
(\ref{eqn:HoffmanWielandtforEigenvalues}), (\ref{eqn:squaresum}), and (\ref{eqn:decayofFouriercoeff}),
the desired result follows for any constant
\begin{equation}
\alpha_{A,\beta} < \gamma \kappa (1-q).
\end{equation}
Moreover, since $\gamma$ can be chosen to be $1$ for $\alpha > 1$ and arbitrary close to $1$ for the von-Neumann entropy, any rate
\begin{equation}
\alpha_{A,\beta} < \kappa (1-q)
\end{equation}
will satisfy the required condition for $\alpha \geq 1$.
Repeating these steps for subset $B$ yields,  for any rate
\begin{equation}
	\alpha_{B,\beta} < \gamma \kappa q,
\end{equation}
the validity of the lemma.
\qed

\begin{lemma}[Error estimate \cite{numericanalysis}]
\label{theo:analyticerror}
Let $g: \rr \longrightarrow \rr$ be real analytical and $2 \pi$-periodic. Then there exists a strip $D=\rr\times\left(-\kappa,\kappa\right) \subset \cc$ with $\kappa >0$ such that $g$ 
can be extended to a holomorphic and $2\pi$-periodic bounded function $g:D \longrightarrow \cc$. The error for the rectangular rule can be estimated by
\begin{equation}
\label{error_equation}
   \left| \frac{1}{2\pi}\int_0^{2\pi}g(x) dx - \frac{1}{n}\sum_{k=1}^n g\left(\frac{2\pi k}{n}\right)\right| \leq \frac{4\pi M}{e^{n\kappa}-1},
\end{equation}
where M denotes a bound for the holomorphic function $g$ on $D$.
\end{lemma}
Note that the constant $\kappa$ will not depend on $n$, however it will clearly depend on the function $g$. It is natural to ask how $\kappa$ will depend on $\beta$ if $M$ is kept at a fixed value when we choose $g$ to be the symbol of the covariance matrix of the thermal state of, say, the 
fermionic instance of the XX model (see Section \ref{sec:temp}). The answer is that $\kappa \propto \beta^{-1}$ for large $\beta$.\\ To see this fix some $\alpha>0$ and define
\begin{equation}
M_{\alpha}:= \sup \left\{\left|\tanh(x) \left|\left|\Im(x)\right| \leq \alpha \frac{\pi}{2}\right.\right|\right\} .
\end{equation}
Note that $M_{\alpha} < \infty$ if and only if $\alpha < 1$ since the hyperbolic tangent has a pole at $\pm \mi \frac{\pi}{2}$. \\
Define the set
\begin{equation}
\tilde{D}_{\alpha}:=\left\{z \in \cc \left| \left|\Im \left(\beta\left(\frac{a}{2}+b \cos(z) \right) \right)\right| \leq \frac{\alpha \pi}{2} \right.\right\}
\end{equation}
then obviously $g(z) \leq M_{\alpha}$ for all $z \in \tilde{D}_{\alpha}$. From $\Im\left(\cos(z)\right) = \sin(\Re(z)) \sinh(-\Im(z))$ follows:
\begin{equation}
D_{\alpha}:= \left\{z \in \cc \left| \left|\Im \left(z\right)\right| \leq \arsinh \left( \frac{\alpha \pi}{2 b \beta}\right) \right.\right\} \subseteq \tilde{D}_{\alpha}
\end{equation}
For fixed $\alpha<1$ set $\kappa:=\arsinh \left( \frac{\alpha \pi}{2 b \beta}\right)$ and $D:=D_{\alpha}$. Then $g$ is 
bounded by  the positive real number $M_{\alpha}$ on $D$. Moreover for any $\alpha \geq 1$, $g$ is necessarily unbounded on $D_{\alpha}$. The asymptopic scaling in $\beta$ follows from
\begin{equation}
\lim_{\beta \rightarrow \infty} \beta \arsinh \left( \frac{\alpha \pi}{2 b \beta}\right) = \frac{\alpha \pi}{2 b}
\end{equation}

\begin{lemma}[Hoffman and Wielandt \cite{Bhatiamatrix}] 
\label{theo:HoffmanWielandt}
Let $A,E \in M_n$ be normal matrices and let $\left(\lambda_1, \lambda_2, \cdots, \lambda_n \right)$ be the eigenvalues of A in some given order and let $\left(\mu_1, \mu_2, \cdots, \mu_n\right)$ be the eigenvalues of $A+E$ in some order. Then there exists a permutation $\sigma \in S_n$ such that
\begin{equation}
    \left(\sum^{n}_{i=1}\left|\mu_{\sigma(i)} - \lambda_i\right|^2\right)^{\frac{1}{2}} \leq \left\|E\right\|_2.
\end{equation}
\end{lemma}

\begin{lemma}[H\"older continuity of entropy functions] 
\label{theo:Hoeldercone}
The (covariance) Renyi entropy functions $s_{\alpha}$ and the (covariance) von-Neumann entropy function $s$ are H{\"o}lder continuous on ${[}-1,1{]}$, i.e., there exist $0 < \gamma \leq 1$ and $C>0$ such that
\begin{equation}
\left|s_{\alpha}(x) - s_{\alpha}(y)\right| \leq C \left|x-y\right|^{\gamma} \textnormal{ for all } x,y \in {[}-1,1{]}.
\end{equation}
Possible H{\"o}lder exponents are
\begin{itemize}
\item $0 < \gamma \leq \alpha$ for $\alpha<1$, 
\item $0 <\gamma<1$ for the von-Neumann entropy and 
\item $0 <  \gamma \leq 1$ for $\alpha>1$.
\end{itemize}
\end{lemma}
\begin{proof}
\hb{
To start with consider the case $1 < \alpha$ first. Note that in this case the first derivative of $s_{\alpha}$ (see Eq.\ (\ref{eqn:firstdersalpha}) in the Appendix) is bounded on the interval ${(}-1,1{)}$. Hence the mean value theorem and boundedness of  $s_{\alpha}'$ imply the existence of some constant $C>0$ such that
\begin{equation}
\frac{\left|s_{\alpha}(x) - s_{\alpha}(y)\right|}{\left|x-y\right|} \leq C  \textnormal{ for all } x,y \in {[}-1,1{]},
\end{equation}
which proves H\"older continuity with H\"older exponent $\gamma = \alpha$ whenever $\alpha>1$. H\"older continuity for $\kappa < \alpha$ folows from 
\begin{equation}
\label{eq:trivialextofHoeldtolowerexp}
  \left|x-y\right|^{\alpha} \leq 2^{\alpha - \kappa} \left|x-y\right|^{\kappa} 
\end{equation}
For $0 < \alpha < 1$ note that for every $0 \leq z \leq 1$ we have
\begin{equation}
\label{eq:helpmeout}
   1 - z^{\alpha} \leq \left(1 - z \right)^{\alpha}.
\end{equation}
By the explicite expression of $s_{\alpha}'$ (given in Eq.\ (\ref{eqn:firstdersalpha}) of the Appendix again) there exists some constant $C>0$ such that
\begin{eqnarray}
\left|s_{\alpha}'(x)\right| \leq \frac{C}{\alpha}\left(\left(1+x\right)^{\alpha-1} + \left(1-x\right)^{\alpha-1}\right).
\end{eqnarray}
This yields for any $-1 \leq x \leq y \leq 1$
\begin{eqnarray}
\label{eq:estimateforhoeldp1}
	\left|s_{\alpha}(x) - s_{\alpha}(y)\right| & = & \left|\int_{x}^{y} s_{\alpha}'(z) dz\right|\\ 
	& \leq & \int_{x}^{y} \frac{C}{\alpha} \left(\left(1+z\right)^{\alpha-1} + \left(1-z\right)^{\alpha-1}\right)  dz \nonumber\\
	& = &  C\left( \left|(1+y)^{\alpha} - (1+x)^{\alpha}\right| + \left|(1-y)^{\alpha} - (1-x)^{\alpha}\right|\right).\nonumber
\end{eqnarray}
If $-1 < y$ then 
\begin{eqnarray}
\label{eq:inequalityhoeld1}
   \left|(1+y)^{\alpha} - (1+x)^{\alpha}\right| & = &  (1+y)^{\alpha}\left(1- \left(\frac{1+x}{1+y}\right)^{\alpha}\right)\\
   & \leq & (1+y)^{\alpha} \left(1- \frac{1+x}{1+y}\right)^{\alpha}\\
   & = & \left|y - x\right|^{\alpha}
\end{eqnarray}
where we used Eq.~(\ref{eq:helpmeout}) in the second line. This relationship remains true for $y=-1$ (and therefore by assumption $-1 \leq x \leq y$ automaticallz $x=-1$). Similarly for all $-1\leq x \leq y \leq 1$:
\begin{eqnarray}
\label{eq:inequalityhoeld2}
   \left|(1-y)^{\alpha} - (1-x)^{\alpha}\right| & \leq &  \left|y - x\right|^{\alpha}
\end{eqnarray}
Hence inequality~(\ref{eq:estimateforhoeldp1}) implies
\begin{eqnarray}
	\left|s_{\alpha}(x) - s_{\alpha}(y)\right| & \leq & 2 C \left|y - x\right|^{\alpha}, \nonumber
\end{eqnarray}
i.e. H\"older continuity with exponent $\kappa=\alpha$. H\"older continuity for $\kappa<\alpha$ follows from Eq.~ (\ref{eq:trivialextofHoeldtolowerexp}) again.\\
The remaining case, $\alpha=1$, can be proven in nearly the same way. By Eq.~(\ref{eqn:firstders}), the derivative of the von Neumann entropy has logarithmic poles at $\pm 1$, and for any $0 \leq \kappa < 1$ can be bounded from above by some positive multiple of the function $g(x):=\left|1+x\right|^{\kappa-1} + \left|1-x\right|^{\kappa-1}$. This is essentially the case considered before with $\alpha$ replaced by $\kappa$. Therefore the von-Neumann entropy is H\"older continuous for any H\"older exponent $\kappa<1$.} \qed
\end{proof}

\subsection{Second order theorem for convex functions of T\"oplitz operators}

In this subsection, we elaborate on second order theorems for convex functions of self-adjoint T\"oplitz operators.
It builds upon the analysis of the main text, but is not directly required in any of the above proofs.
For the main statement of this subsection, we will need Peierls' inequality in the following form.

\begin{lemma}[Peierl's inequality]
\label{lem:Peierl}
Let $A=A^{\dagger}$ be a Hermitian $N \times N$ matrix and let $U$ be a unitary $N \times N$ matrix. Then, 
for any function $f$ that is convex on some open interval $I$ with $\sigma{A} \in I$, one finds that
\begin{equation}
  \sum_{1 \leq i \leq N}f \left( \left(U A U^{\dagger}\right)_{i,i} \right) \leq \tr f(A).
\end{equation}
\end{lemma}
The following statement holds.

\begin{theorem}[Self-adjoint T\"oplitz operators]
\label{th:Szegoselfadjoint}
Let $\mu \in W \cap B^{\frac{1}{2}}_2$ be a symbol. Assume further that $\bild(\mu) \subseteq \rr$ (or equivalently assume that $T(\mu)=T(\mu)^{\dagger}$). Then the associated family of T\"oplitz matrices, $T_n\in\cc^{n\times n}$ is Hermitian. Let $J \subseteq \rr$ be an open interval with $\bild(\mu) \subseteq J$ and let $g: J \rightarrow \rr$ be real- analytic. Then there exists an analytic extension $\tilde{h}$ of $h$ to a small complex neighborhood of $\tilde{J}$ of $\bild(\mu)$ with $\tilde{g}\left|_{\bild(\mu)}\right. = g\left|_{\bild(\mu)}\right.$ (compare Lemma~\ref{theo:analyticerror}) and the trace formula of Theorem~\ref{th:Szego} \hb{and the alternative calculation formula of Theorem~\ref{th:Widomformula}} hold true for this extension. 
Moreover if $g$ is convex, then
\begin{equation}
 E_g(\mu) \leq 0
\end{equation}
\end{theorem}
\begin{proof}
The only non-obvious statement is 
\begin{equation}
 E_g(\mu) \leq 0
\end{equation}
Fix $\varepsilon>0$ and $n_0 \in \nn$ large enough such that
\begin{equation}
 \tr g(T_n) = n G_g(\mu) + E_g(\mu) + r_n
\end{equation}
with $\left|r_n\right|< {\varepsilon}/{3}$ for every $n \geq n_0$. Since $T_n$ is Hermitian we have $r_n \in \rr$. Choose unitary matrices $U$ that diagonalize $T_{n_0}$. Then
\begin{equation}
   2 g(T_{n_0}) = \sum_{1 \leq i \leq 2 n_0} g\left( \left(U \oplus U T_{2 n_0} U^{\dagger} \oplus U^{\dagger}\right)_{i,i}\right) = 2 n_0 G_g(\mu) + 2 E_g(\mu) + 2 r_{n_0} .
\end{equation}
Hence, by Peierl's inequality:
\begin{equation}
  2 n G_g(\mu) + 2 E_g(\mu) + 2 r_{n_0} \leq \tr \left(g\left(T_{2 n_0}\right)\right) = 2 n_0 G_g(\mu) + E_g(\mu) + r_{2 n_0} .
\end{equation}
Therefore:
\begin{equation}
E_g(\mu) \leq r_{2 n_0} - 2 r_{n_0} < \varepsilon .
\end{equation}
Since $\varepsilon>0$ was arbitrary this implies
\begin{equation}
E_g(\mu)\leq 0.
\end{equation}\qed
\end{proof}
Since obviously $E_{g_1+g_2}(\mu)= E_{g_1}(\mu)+E_{g_2}(\mu)$ and $E_{\lambda g}(\mu) = \lambda E_{g}(\mu)$ for $\lambda \in \cc$ this implies that $E_{g_1}(\mu) \leq E_{g_2}(\mu)$ whenever $g_1 - g_2$ is convex. If both $g_1$ and $g_2$ are differentiable twice this holds true if and only if $g_1'' \leq g_2''$.

\subsection{Proof of Theorem \ref{theo:tempdep}}\label{tempproof}

\begin{proof}
We start with proving the part of the theorem relating to the high temperature dependence,
Eq.\ (\ref{eq:hightemplimit})). For that purpose, we formulate the 
explicit asymptotic expression for the mutual information of a thermal state of the fermionic instance of the XX model. 
Define
\begin{equation}
\epsilon_{\beta} (\phi):= \beta \left(\frac{a}{2} + b\cos(\phi)\right).
\end{equation}
Then
\begin{equation}
	\epsilon'_{\beta} (\phi):=- \beta b \sin(\phi).
\end{equation}
The symbol takes the 
form $\lambda(\me^{\mi\phi})= -\tanh(\epsilon_{\beta}(\phi))$. In these terms,
we can express the mutual information as
\begin{equation}
I_{\alpha,\infty}(A:B)=\frac{1}{4\pi^2}\int_{-\pi}^{\pi}\int_{-\pi}^{\pi}d\theta d\phi f_\beta(\theta,\phi)t_\beta(\theta,\phi),
\end{equation}
where we have decomposed the integrand into the functions
\begin{eqnarray}
f_\beta(\theta,\phi)&:=&\frac{s_{\alpha}(\tanh(\epsilon_{\beta}(\theta)))-s_{\alpha}(\tanh(\epsilon_{\beta}(\phi)))}{\tanh(\epsilon_{\beta}(\theta))-\tanh(\epsilon_{\beta}(\phi))}
	\\
	&\times &
	\left(\frac{{\epsilon'_{\beta}}(\theta)}{\cosh^2{(\epsilon_{\beta}(\theta))}} 
	- \frac{{\epsilon'_{\beta}}(\phi)}{\cosh^2{(\epsilon_{\beta}(\phi))}}\right) ,\nonumber\\
	t_\beta(\theta,\phi)&:=&\tan^{-1}\left(\frac{\theta-\phi}{2}\right).
\end{eqnarray}
The high temperature (small $\beta$) limit can be calculated simply by expanding the functions in the integrand, to leading non-zero order in $\beta$, which yields
\begin{equation}
f_\beta(\theta,\phi)= \frac{\alpha b \beta^2 }{2 \log(2)}\left(a+b \left(\cos{\theta}+\cos{\phi}\right)\right)(\sin{\theta}-\sin{\phi}) +O(\beta^3).
\end{equation}
Now, noting that 
\begin{equation}
\int_{-\pi}^{\pi}\int_{-\pi}^{\pi}d\theta d\phi \left(a+b \left(\cos{\theta}+\cos{\phi}\right)\right)(\sin{\theta}-\sin{\phi}) \tan^{-1} \left((\theta-\phi)/2 \right) = 4 b \pi^2,
\end{equation}
we find that 
\begin{equation}
  \lim_{\beta \rightarrow 0} \frac{I_{\alpha,\infty}(A:B)}{\beta^2}=\frac{\alpha b^2}{2 \log(2)} .
\end{equation}
\qed
\end{proof}

The second part of the theorem, relating to the 
low temperature asymptotics (i.e., the asymptotics for large $\beta$), turns out to be significantly 
more involved. In a first step, we will cast the integral into a more appropriate form. By an elementary argument
involving the symmetry of the problem, one finds that
\begin{equation}
\label{eq:muualinformationbyI1I2}
 	I_{\alpha,\infty}(A:B) = \frac{1}{2 \pi^2} \left(I_{1} + I_{2}\right),
\end{equation}
where
\begin{equation}
	I_{1}= \int_{0}^{\pi}\int_{0}^{\pi}d\theta d\phi f_\beta(\theta,\phi)t_\beta(\theta,\phi),
\end{equation}
and
\begin{equation}
	I_{2}=  \int_{0}^{\pi}\int_{-\pi}^{0}d\theta d\phi f_\beta(\theta,\phi)t_\beta(\theta,\phi),
\end{equation}
In this expression, we have made use of the same notation as in the above proof 
of the low temperature behavior. Note that both $\cos$ functions 
are monotonous functions of there argument in both regions of integration. The substitution 
\begin{eqnarray}
	x &:= & \cos(\theta) ,\\
	y &:=& \cos(\phi). 
\end{eqnarray}
and an appropriate rearrangement of factors then gives
\begin{equation}
\label{eqn:I1withffunctions}
I_{1}= \beta b \int_{-1}^{1}\int_{-1}^{1}dx dy f_{\beta,11}(x,y) \left(f_{\beta,12}(x,y) + f_{\beta,13}(x,y)\right),
\end{equation}
where 
\begin{eqnarray}
f_{\beta,11}(x,y)&:=& \frac{s_{\alpha}(\tanh(\tilde{\epsilon}_{\beta}(x)))-s_{\alpha}(\tanh(\tilde{\epsilon}_{\beta}(y)))}{x-y},\\
f_{\beta,12}(x,y)&:=& \frac{\sech^2(\tilde{\epsilon}_{\beta}(x))\left(1-x^2\right) -\sech^2(\tilde{\epsilon}_{\beta}(y))\left(1-y^2\right)}{\tanh(\tilde{\epsilon}_{\beta}(x))- \tanh(\tilde{\epsilon}_{\beta} (y))} \frac{1}{({1-x^2})^{1/2}({1-y^2})^{1/2}},\nonumber\\
 \\
f_{\beta,13}(x,y)&:=& \frac{\sech^2(\tilde{\epsilon}_{\beta}(x)) -\sech^2(\tilde{\epsilon}_{\beta}(y))}{\tanh(\tilde{\epsilon}_{\beta}(x))- \tanh(\tilde{\epsilon}_{\beta} (y))}.  
\end{eqnarray}
In this equation, we have defined
\begin{equation}
	\tilde{\epsilon}_{\beta}(x):= \beta ({a}/{2} + bx).
\end{equation}
An analogous computation for $I_2$ yields:
\begin{equation}
I_{2}=\beta b \int_{-1}^{1}\int_{-1}^{1}dx dy f_{\beta,11}(x,y) \left(f_{\beta,12}(x,y) - f_{\beta,13}(x,y)\right).
\end{equation}
We need the following elementary but powerful lemma.

\begin{lemma}[Bounds on difference quotients]\label{lem:boundsondifquot}
Let $f,g,h,k$ be differentiable functions on some open interval $I \subseteq \rr$. Assume
\begin{itemize}
\item $s'(z) \geq 0$ for any  $s \in \left\{f,g,h,k\right\}$ and any $z \in I$,
\item $\sup \left\{s'(z) \left| z \in I\right.\right\} < \infty$ for any $s \in \left\{f,g,h,k\right\}$,
\item $f'(z) \leq h'(z)$ and $k'(z) \leq g'(z)$ for all $z \in I$.
\end{itemize}
Then, for any pair of real numbers $x,y \in I$ with $x \neq y$,
\begin{equation}
0 \leq \frac{f(x) - f(y)}{g(x) - g(y)} \leq \frac{h(x) - h(y)}{k(x) - k(y)}.
\end{equation}
\end{lemma}
\begin{proof}
By the mean value theorem there exists some $z \in I$ with $x <z<y$ such that 
\begin{equation}
\frac{(h-f)(x) - (h-f)(y)}{x-y} = h'(z) - f'(z) \geq 0.
\end{equation}
Hence 
\begin{equation}
\frac{h(x) - h(y)}{x-y} \geq  \frac{f(x) - f(y)}{x-y} \geq 0,
\end{equation}
where the last inequality follows from $f' \geq 0$. By the same argument
\begin{equation}
\frac{g(x) - g(y)}{x-y} \geq  \frac{k(x) - k(y)}{x-y} \geq 0,
\end{equation}
multiplying both inequalities and rearranging terms yields the desired lemma.
\qed
\end{proof}
Moreover, we will use the following integral.

\begin{lemma}[Integral over exponential difference quotient]\label{lem:Eiintegral}
For any $a_1<a_2$ and $b_1<b_2$,
\begin{eqnarray}
A_{a_1,a_2,b_1,b_2} & := & \int^{a_2}_{a_1} \int^{b_2}_{b_1} dx dy \frac{\exp(-x) - \exp(-y)}{y - x}  \\
& = & F(b_1,a_1,a_2) + F(a_1,b_1,b_2) - F(b_2,a_1,a_2) - F(a_2,b_1,b_2). \nonumber
\end{eqnarray}
Here
\begin{equation}
F(a,b,c):=\exp(-a)\left[\tilde{F}(b-a) - \tilde{F}(c-a)\right]
\end{equation}
with
\begin{equation}
\tilde{F}(x):= \Ei[-x]-\log\left|x\right| = \lambda_E + \sum_{k \geq 1} \frac{(-1)^{k}z^k}{k k!}
\end{equation}
where
\begin{equation}
\Ei(x):=\int^{x}_{-\infty} \frac{\exp(t)}{t} 
\end{equation}
is the exponential integral function and $\gamma_E$ denotes the Euler Gamma constant. 
\end{lemma}

\begin{figure}
\begin{center}
\begin{tikzpicture}[scale=2]
\draw [black] (0,0) rectangle (2,2);
\draw [fill=black] (0.25,0.25) rectangle (0.75,0.75);
\draw [fill=gray] (0,0.25) rectangle (0.25,0.75);
\draw [fill=gray] (0.75,0.25) rectangle (2,0.75);
\draw [fill=gray] (0.25,2) rectangle (0.75,0.75);
\draw [fill=gray] (0.25,0) rectangle (0.75,0.25);
\draw [fill=black] (2.5,1.9) rectangle (2.9,1.5);
\node at (3.1,1.7) {\normalsize $Q_{\delta}$};
\draw [black,fill=white] (2.5,1.3) rectangle (2.9,0.9);
\node at (3.1,1.1) {\normalsize $Q_{e,\delta}$};
\draw [fill=gray] (2.5,0.7) rectangle (2.9,0.3);
\node at (3.1,0.5) {\normalsize $Q_{r,\delta}$};
\draw [dotted,thick] (2,-0.1)--(2,0);
\draw [dotted,thick] (0,-0.1)--(0,0);
\draw [dotted,thick] (-0.1,2)--(0,2);
\draw [dotted,thick] (-0.1,0)--(0,0);
\node at (-0.3,2) {\footnotesize $1$};
\node at (-0.3,0) {\footnotesize $-1$};
\node at (2,-0.2) {\footnotesize $1$};
\node at (0,-0.2) {\footnotesize $-1$};
\draw [dotted,thick] (-0.1,0.5)--(2,0.5);
\draw [dotted,thick] (-0.1,1)--(2,1);
\draw [dotted,thick] (-0.1,1.5)--(2,1.5);
\node at (-0.3,0.5) {\footnotesize  $-0.5$};
\node at (-0.3,1) {\footnotesize $0$};
\node at (-0.3,1.5) {\footnotesize $0.5$};
\draw [dotted,thick] (0.5,-0.1)--(0.5,2);
\draw [dotted,thick] (1,-0.1)--(1,2);
\draw [dotted,thick] (1.5,-0.1)--(1.5,2);
\node at (0.5,-0.2) {\footnotesize $-0.5$};
\node at (1,-0.2) {\footnotesize $0$};
\node at (1.5,-0.2) {\footnotesize  $0.5$};
\end{tikzpicture}
\end{center}
\caption{The different regions of integration for $a=b$ and $\delta=0.5$.}\label{temp}
\end{figure}
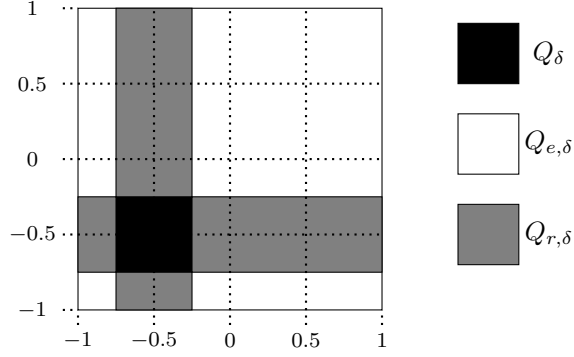

Having established these two lemmas we proceed with the proof of Theorem \ref{theo:tempdep}:
\begin{proof}(Of Theorem~\ref{theo:tempdep}, Eq.\ (\ref{eq:lowtemplimit1}) and Eq. (\ref{eq:lowtemplimitferro})).
Consider the open interval 
\begin{equation}
	I:=\left(\frac{a}{2}-b,\frac{a}{2}+b\right).
\end{equation}
If $0 \in I$ (i.e., $\left|a\right| < 2 b$), set $m:=1-|a|/(2b)$. If $0 \notin \overline{I}$,
set $m:= |a|/(2b)-1$.
For $1>\delta>0$ we will split the region of integration into several sets, namely the following.
\begin{itemize}
\item  The bulk contribution: 
\begin{equation}
Q_{\delta}:=\left\{(x,y) \in \left[-1,1\right] \times \left[-1,1\right] \left|\left|x+\frac{a}{2b}\right|,\left|y+\frac{a}{2b}\right| \leq \delta m\right.\right\}.
\end{equation}
If $\left|a\right|>2b$, then $Q_{\delta}:=\emptyset$ for every $0 < \delta < 1$.
\item The edge contribution: 
\begin{equation}
Q_{e,\delta}:=\left\{(x,y) \in \left[-1,1\right] \times \left[-1,1\right] \left|\left|x+\frac{a}{2b}\right|,\left|y+\frac{a}{2b}\right|> \delta m\right.\right\}.
\end{equation}
If $\left|a\right|>2b$, then $Q_{\delta}= \left[-1,1\right] \times \left[-1,1\right]$
\item The remaining contribution: 
\begin{equation}
Q_{r,\delta}:=\left[-1,1\right]\times \left[-1,1\right] \setminus \left(Q_{\delta} \cup Q_{e,\delta} \right).
\end{equation}
\end{itemize}
The various regions of integration are depicted in Fig.\ \ref{temp}.
We will prove the theorem in several steps:
\begin{itemize}
\item[{\bf Step 1:}] {\it The $Q_{e,\delta}$ contribution of $I_1$ and $I_2$ decays exponentially fast}.
Formally, we will show that for any $0 < \delta <1$ and every $\kappa<2 \delta b m \min\left\{\alpha,1\right\}$
\begin{equation}
\label{eq:fundamentaltailestimate}
	\lim_{\beta \rightarrow \infty} \exp(\kappa \beta) 
	\int \int_{Q_{e,\delta}} dxdy f_{\beta,11}(x,y) \left(f_{\beta,12}(x,y) \pm  f_{\beta,13}(x,y) \right) = 0 ,
\end{equation}
where we used the same notation as in Eq.\ (\ref{eqn:I1withffunctions}) (note that this proves Eq.\ (\ref{eq:lowtemplimitferro})). 
First of all observe that
\begin{equation}
 f_{\beta,11} \left(f_{\beta,12} \pm f_{\beta,13}\right) = \left|f_{\beta,11}\right|\cdot\left|f_{\beta,12} \pm f_{\beta,13}\right|. 
\end{equation}
Consider the factor $\left|f_{\beta,11}\right|$ first. 
If $x,y > {-a}/{2b}+\delta$ or $x,y < {-a}/{2b}-\delta$ (this is actually true for all $x,y$ whenever $|a| > 2b$) the mean value theorem yields
\begin{eqnarray}
	\left|f_{\beta,11}\right| & = & \beta b \cdot \sech^2\left(\tilde{\epsilon}_{\beta}(z) \right) 
	\cdot  s_{\alpha}'\left(\tanh\left(\tilde{\epsilon}_{\beta}(z)\right)\right) \nonumber \\
	& \lesssim & (\beta b)^2 C_{1} \exp(-2\min\left\{\alpha,1\right\}\beta b \delta m ).
\end{eqnarray}
Here $z$ lies between $x$ and $y$ and $C_1$ is some constant (any $C_1 > {4 \alpha}/{\log(2)}$ works). 
Here and in the future, we will use the symbol $\lesssim$ to indicate that $\leq$ holds for sufficiently large $\beta$. 

If $\left|a\right|<2b$ then $x + a/(2b)< - \delta, y + a/(2b)> \delta$ has to be estimated separately. In this case consider the estimate
\begin{equation}
 \left|f_{\beta,11}\right| \leq \frac{2 \sup\left\{s_{\alpha}(z)\left|\left|z\right|\geq \tanh(\delta \beta b m)\right.\right\}}{2 \delta} \lesssim C_{2,\delta,\alpha} \beta b \exp(-2 \beta m \delta b \alpha),
\end{equation}
for some (possibly $\delta-$ and $\alpha-$ dependent) constant $C_{2,\delta,\alpha}$ (the factor $\beta b$ is actually only needed if $\alpha=1$). So asymptotically, we have that
\begin{equation}
   \left|f_{\beta,11}\right| \lesssim  (C_1+C_{2,\delta,\alpha})(b \beta)^2\exp(-2\min\left\{\alpha,1\right\}\beta b \delta m ) .
\end{equation}
Consequently, setting $C_2:=C_1+C_{2,\delta,\alpha}$, we obtain
\begin{eqnarray}
\label{Eq.firstfactorexpsmall}
	&& \int_{Q_{e,\delta}} dxdy f_{\beta,11}(x,y) (f_{\beta,12}(x,y) + f_{\beta,13}(x,y)) \\
	& \leq & C_{2} (b \beta)^2 \exp(-2 \beta b \delta m \min\left\{\alpha,1\right\}) J \nonumber,
\end{eqnarray}
where
\begin{equation}
J = \int_{\left[-1,1\right] \times \left[-1,1\right]} dxdy \left|f_{\beta,12}(x,y)\right| + \left|f_{\beta,13}(x,y)\right|  .
\end{equation}
The generalized mean value theorem and some calculus yields
\begin{equation}
\left|\frac{\sech^2(\tilde{\epsilon}_{\beta}(x)) - \sech^2(\tilde{\epsilon}_{\beta}(y))}{\tanh(\tilde{\epsilon}_{\beta}(x)) - \tanh(\tilde{\epsilon}_{\beta}(x))} \right| \leq 2
\end{equation}
and
\begin{equation}
 \left|\frac{\sech^2(\tilde{\epsilon}_{\beta}(x))(1-x^2) - \sech^2(\tilde{\epsilon}_{\beta}(y))(1-y^2)}{\tanh(\tilde{\epsilon}_{\beta}(x)) - \tanh(\tilde{\epsilon}_{\beta}(x))} \right|\leq 2 + \frac{2}{b \beta} .
\end{equation}
Hence 
\begin{equation}
  J \lesssim (2+\epsilon) \pi^2
\end{equation}
for any $\epsilon>0$. This together with Eq.\ (\ref{Eq.firstfactorexpsmall}) 
yields the identity Eq.\ (\ref{eq:fundamentaltailestimate}) and finishes the first step. 

\item[{\bf Step 2}:] {\it The $Q_{r,\delta}$ contribution of $I_1$ and $I_2$ is bounded}.
Formally we will show that
\begin{equation}
\label{eq:fundamentaltailestimaterem}
\limsup_{\beta \rightarrow \infty} \beta b \int \int_{Q_{r,\delta}} dxdy f_{\beta,11}(x,y) \left(f_{\beta,12}(x,y)\pm f_{\beta,13}(x,y)\right) < \infty.
\end{equation}
Without any loss of generality we can assume that $\delta <{1}/{4}$. We will will split ${Q_{r,\delta}}$ into the sets
\begin{equation}
  {Q_{r,\delta,1}}:={Q_{r,\delta}} \cap \left\{(x,y) \in \left[-1,1\right] \times \left[-1,1\right]\left|\left|x-y\right|<\delta\right.\right\}
\end{equation}
and 
\begin{equation}
  {Q_{r,\delta,2}}:= {Q_{r,\delta,1}} \setminus {Q_{r,\delta,1}} .
\end{equation}
We will show separately that the two integrals 
\begin{equation}
J_1:=\beta b \int \int_{Q_{r,\delta,1}} dxdy \left|f_{\beta,11}(x,y)\right| \left(\left|f_{\beta,12}(x,y)\right|+ \left|f_{\beta,13}(x,y)\right|\right),
\end{equation}
and
\begin{equation}
J_2:=\beta b \int \int_{Q_{r,\delta,2}}dxdy \left|f_{\beta,11}(x,y)\right| \left(\left|f_{\beta,12}(x,y)\right|+ \left|f_{\beta,13}(x,y)\right|\right)
\end{equation}
are bounded. In order to bound $J_1$, note that by the mean value theorem 
\begin{equation}
	\left(\left|f_{\beta,12}(x,y)\right|+ \left|f_{\beta,13}(x,y)\right|\right) 
	\leq \left(\left(2+\frac{2}{\beta b}\right)\frac{1}{{(1-4\delta^2)^{1/2}(1-\delta^2)^{1/2}}}+2\right).
\end{equation}
Hence,
\begin{equation}
 J_2 \leq C_{\delta} \beta b \int \int_{{Q_{r,\delta,1}}}dxdy \left|f_{11}(x,y)\right| .
\end{equation}
An appropriate linear substitution gives
\begin{eqnarray}
 && \beta b \int \int_{{Q_{r,\delta,1}}}dxdy \left|f_{11}(x,y)\right| \\
 & = & 2 \int^{2\delta b \beta}_{\delta b \beta} \int^{\delta 
\beta b}_{0} dx dy \left|\frac{s_{\alpha}(\tanh{(x)})-s_{\alpha}(\tanh{(y)})}{x-y}\right| .
\end{eqnarray}
Note that the function $f_{\alpha}:\left[0,\infty\right)\rightarrow\rr$ 
defined as $f_{\alpha}(x):=s_{\alpha}(\tanh(x))$ is monotoneous decreasing. 
Moreover,
\begin{eqnarray}
\label{eq:falphadashexp}
&& -f_{\alpha}'(x) = \\
&& \frac{2\alpha}{\left|1-\alpha\right|\log(2)}\exp\left(-2 \min\left\{1,\alpha\right\} x \right) \frac{1-\exp(-2\left|\alpha-1\right|x)}{1+\exp(-2\alpha x)}\frac{1}{1+\exp(-2x)} \nonumber
\end{eqnarray}
for $\alpha \neq 1$ and 
\begin{equation}
\label{eq:fonedashexp}
 -f_{1}'(x) =  \frac{4}{\log(2)}  \frac{x \exp\left(-2 x \right)}{\left(1+\exp(-2x)\right)^2} .
\end{equation}
Therefore, 
there exist  constants $C_{\alpha},\kappa_{\alpha}>0$ such that 
\begin{equation}
\label{eq:estimderfalpha}
 -f_{\alpha}'(x) \leq C_{\alpha} \exp(-\kappa_{\alpha} x).
\end{equation}
That is to say, Lemma~\ref{lem:boundsondifquot} gives
\begin{equation}
\label{eqn:roughestimateonsdifquot}
	 0\leq \frac{s_{\alpha}(\tanh{(x)})-s_{\alpha}(\tanh{(y)})}{y-x} \leq \frac{C_{\alpha}}{\kappa_{\alpha}} \frac{\exp(-\kappa_{\alpha} x) - \exp(-
	 \kappa_{\alpha} y)}{y-x},
\end{equation}
giving rise to the estimate
\begin{eqnarray}
	 && \int^{2\delta b \beta}_{\delta b \beta} \int^{\delta 
	\beta b}_{0} dx dy \left|\frac{s_{\alpha}(\tanh{(x)})-s_{\alpha}(\tanh{(y)})}{x-y}\right| \\
	&\leq& \frac{C_{\alpha}}{(\kappa_{\alpha})^2} 
	A_{\delta b \beta \kappa_{\alpha},2 \delta b \beta \kappa_{\alpha},0,\delta b \beta \kappa_{\alpha}} \nonumber ,
\end{eqnarray}
where $A_{a_1,a_2,b_1,b_2}$ has been defined in Lemma \ref{lem:Eiintegral}. By some calculus again
\begin{equation}
\lim_{\beta \rightarrow \infty}  A_{\delta b \beta \kappa_{\alpha},2 \delta b \beta \kappa_{\alpha},0,\delta b \beta \kappa_{\alpha}} = \log \left(2 \right) .
\end{equation}
Therefore,
$\limsup_{\beta \rightarrow \infty} J_1 < \infty$.
In order to bound $J_2$,
 the singularities of $ f_{13}$ on the lines $x=1$ and $y=1$ have to be controlled. We use the bound 
\begin{equation}
 \left|f_{1,2}(x,y) + f_{1,3}\right| \leq \left(4+\frac{2}{b\beta} \right)\frac{1}{({1-x^2})^{1/2}({1-y^2})^{1/2}}
\end{equation}
again and bound $f_{1,1}$ in the following way:
\begin{equation}
\left|f_{1,1}(x,y)\right| \leq \frac{s_{\alpha}(\tanh(\tilde{\epsilon}_{\beta}(x))) + s_{\alpha}(\tanh(\tilde{\epsilon}_{\beta}(y))}{\delta} .
\end{equation}
A short calculation shows that
\begin{equation}
	J_2 \lesssim 8 \frac{4+\frac{2}{\beta b}}{\delta} (\beta b)^2 C_{\alpha} \int^1_0 \int^1_0 x \exp(-2 \beta \alpha b x) 
	\left(\frac{1}{({1-x^2})^{1/2}({1-y^2})^{1/2}} \right),
\end{equation}
where $ C_{\alpha}$ is some constant with $ C_{\alpha} > {4\alpha}/{\log(2)}$. Note that
\begin{equation}
	\label{eg:boundonsqrtintegrals}
	 \limsup_{\beta \rightarrow \infty } \beta^2 \int^1_0 \int^1_0 dx dy x 
	 \exp(-2 \beta b \alpha x) \frac{1}{({1-x^2})^{1/2}} \frac{1}{({1-y^2})^{1/2}} \leq \frac{\pi}{2(2\alpha b)^2} .
\end{equation}
This follows easily from explicit integration of the $y-$ integral, giving 
\begin{equation}
\int^1_0 dy \frac{1}{({1-y^2})^{1/2}}=\frac{\pi}{2}
\end{equation}
and an appropriate estimate for the $x-$ integral, namely for fixed $\epsilon > 0$ note that
\begin{eqnarray} 
 && \int^1_{0} dx x \exp(-2 \beta b \alpha x) \frac{1}{({1-x^2})^{1/2}} \\
 & \leq & \frac{1}{({1-\epsilon^2})^{1/2}}\int^{\epsilon}_{0} dx x \exp(-2 \beta b \alpha x) + \exp(-2 \beta b \alpha \epsilon) \int^{1}_{\epsilon} 
 \frac{1}{({1-x^2})^{1/2}} \nonumber\\
&\leq& \frac{1}{(2\alpha b \beta)^2 ({1-\epsilon^2})^{1/2}} + \exp(-\epsilon \beta b) \frac{\pi}{2} .\nonumber
\end{eqnarray}
Finally, multiply this expression by $\beta^2$, take the $\limsup$ for $\beta$ to infinity and let $\epsilon$ approach zero afterwards (bounding 
$1/(1-x^2)^{1/2}$ by $1$ from below actually shows that Eq.\ (\ref{eg:boundonsqrtintegrals}) 
holds true with equality and that the limit superior is actually a limit). 
In any case $J_2$ is bounded, proving the second step.

\item[{\bf Step 3:}] {\it The integrand $f_{\beta,11}$ bounds the asymptotic growth of $I_1$ and $I_2$}.
Formally we will show that 
\begin{eqnarray}
\label{eqn:asymptoticallyimportantpart1}
\limsup_{\beta \rightarrow \infty} \frac{I_1}{\log(\beta)} & \leq  &  \liminf_{\beta \rightarrow \infty}  \frac{4 \beta b  \int_{Q_{\delta}} dx dy \left|f_{\beta,11}(x,y)\right|}{\log(\beta)}
\end{eqnarray}
and
\begin{eqnarray}
	\label{eqn:asymptoticallyimportantpart2}
	\limsup_{\beta \rightarrow \infty} \frac{I_2}{\log(\beta)} & \leq  &  \liminf_{\beta \rightarrow \infty}  \frac{4 \beta b  \int_{Q_{\delta}} dx dy 
	\left|	f_{\beta,11}(x,y)\right|}{\log(\beta)}
\end{eqnarray}for any $0<\delta<1$. 
The result for $I_2$ is trivial when the result for $I_1$ is settled, since $I_2 \leq I_1$.
Fix $\epsilon>0$ and note that by Step 1 and Step 2, we find that
\begin{equation}
	\limsup_{\beta \rightarrow \infty} \frac{I_1}{\log(\beta)}  = \limsup_{\beta \rightarrow \infty} \frac{I_{1,\epsilon}}{\log(\beta)} 
\end{equation}
where
\begin{equation}
	I_{1,\epsilon}:=b \beta \int \int_{Q_{\epsilon}}dx dy f_{11}(x,y) \left(f_{12}(x,y)+f_{13}(x,y) \right).
\end{equation}
However, for $x,y \in Q_{\epsilon}$ the mean value theorem yields the estimate
\begin{equation}
  \left|f_{12}(x,y)+f_{13}(x,y)\right| \leq 2 + 2\left(1+\frac{1}{\beta b} \right)\frac{1}{1-\epsilon^2}
\end{equation}
again (compare Step 1 and Step 2). This is especially true for the choice $0< \epsilon < \delta$, and therefore
\begin{eqnarray}
\limsup_{\beta \rightarrow \infty} \frac{I_{1}}{\log(\beta)} & \leq  &  \liminf_{\beta \rightarrow \infty}  \left(2+2\frac{1}{1-\epsilon^2}\right) \frac{ \beta b \int_{Q_{\epsilon}} dx dy \left|f_{\beta,11}(x,y)\right|}{\log(\beta)}\nonumber \\
& \leq  &  \liminf_{\beta \rightarrow \infty}  \left(2+2\frac{1}{1-\epsilon^2}\right) \frac{ \beta b \int_{Q_{\delta}} dx dy \left|f_{\beta,11}(x,y)\right|}{\log(\beta)}.
\end{eqnarray}
In the last expression $\epsilon$ can be sent to zero, proving the upper bound, Eq.\ (\ref{eqn:asymptoticallyimportantpart1}).

\item[{\bf Step 4:}] {\it Derive a logarithmic bound for $I_1$ and show that $I_2$ increases sub-logarithmically}.
We will show that 
\begin{equation}
\limsup_{\beta \rightarrow \infty} \frac{I_1}{\log(\beta)} \leq \min\left\{\frac{108 \alpha}{\log(2) \min\left\{1,\alpha\right\}^3},\frac{16 \alpha}{\log(2) \left|1-\alpha\right| \min\left\{1,\alpha\right\}^2}\right\}
\end{equation}
and
\begin{equation}
\label{eq:discussI2away}
\limsup_{\beta \rightarrow \infty} \frac{I_2}{\log(\beta)} = 0.
\end{equation}
We use Eq.\ (\ref{eqn:asymptoticallyimportantpart1}) from Step 3, with which we make a linear substitution and exploit the
symmetry of the integrand to get
\begin{eqnarray}
\limsup_{\beta \rightarrow \infty} \frac{I_1}{\log(\beta)} & \leq  &  
\limsup_{\beta \rightarrow \infty}   \frac{8\left(J_{11}+J_{12}\right)}{\log(\beta)},
\end{eqnarray}
where
\begin{equation}
J_{11}=\int^{\delta \beta b}_{0} \int^{\delta \beta b}_{0} dx dy \frac{f_{\alpha}(x) - f_{\alpha}(y)}{y-x}
\end{equation}
and
\begin{equation}
J_{12}=\int^{\delta \beta b}_{0} \int^{\delta \beta b}_{0} dx dy \frac{f_{\alpha}(x) + f_{\alpha}(y)}{x+y}.
\end{equation}
Here, $f_{\alpha}$ are the functions introduced in Step 2. 
Fix $0<\epsilon<\min\left\{\alpha,1\right\}$. By Eq.\ (\ref{eq:estimderfalpha}),
one finds
\begin{equation}
  -{f_{\alpha}}'(x) \leq C_{\alpha} \exp(-\kappa_{\alpha}x).
\end{equation}
A look at Eq.\ (\ref{eq:estimderfalpha}) and Eq.\ (\ref{eq:fonedashexp}) in turn reveals that
\begin{equation}
C_{\alpha} := \frac{2 \alpha}{\log(2) \epsilon}, \kappa_{\alpha}:=2(\min\left\{\alpha,1\right\}-\epsilon)
\end{equation}
is a possible choice. Proceeding as in Step 3 yields
\begin{equation}
 \limsup_{\beta \rightarrow \infty} \frac{J_{11}}{\log(\beta)} \leq \limsup_{\beta \rightarrow \infty} \frac{C_{\alpha}}{(\kappa_{\alpha})^2} \frac{A_{0, \delta b \beta \kappa_{\alpha},0,\delta b \beta \kappa_{\alpha}}}{\log(\beta)}=\frac{2 C_{\alpha}}{(\kappa_{\alpha})^2}
\end{equation}
and optimizing for $\epsilon$ gives
\begin{equation}
 \limsup_{\beta \rightarrow \infty} \frac{J_{11}}{\log(\beta)} \leq \frac{27 \alpha}{8\log(2) \left(\min\left\{1,\alpha\right\}\right)^3}.
\end{equation}
Another possibility (working for $\alpha \neq 1$ only) is to use the bound 
\begin{equation}
-{f_{\alpha}}'(x) \leq \frac{2 \alpha}{\left|1-\alpha\right|\log(2)} \exp(-2 \min\left\{a,1\right\}x) ,
\end{equation}
giving the asymptotic bound
\begin{equation}
 \limsup_{\beta \rightarrow \infty} \frac{J_{11}}{\log(\beta)} \leq \frac{\alpha}{\left|1-\alpha\right|\log(2) \left(\min\left\{a,1\right\}\right)^2}.
\end{equation}
To bound $J_{12}$, note that $x+y \geq \left|x-y\right|$ for any $x,y\geq 0$. Therefore,
\begin{equation}
	J_{12}\leq J_{11}.
\end{equation}
To show the validity of Eq.\ (\ref{eq:discussI2away}), we
use the mean value theorem to yield the following bounds
\begin{equation}
	\left|\frac{\sech^2(\tilde{\epsilon}_{\beta}(x)) - 
	\sech^2(\tilde{\epsilon}_{\beta}(y))}{\tanh(\tilde{\epsilon}_{\beta}(x)) - 
	\tanh(\tilde{\epsilon}_{\beta}(x))}\left(\frac{1}{\sqrt{1-x^2}\sqrt{1-y^2}}-1\right) \right| \leq 2\left(\frac{1}{1-\delta^2} -1 \right)
\end{equation}
for all $x,y \in Q_{\delta}$. By the same argument,
\begin{equation}
	\left|\frac{x^2\sech^2(\tilde{\epsilon}_{\beta}(x)) - 
	y^2\sech^2(\tilde{\epsilon}_{\beta}(y))}{\tanh(\tilde{\epsilon}_{\beta}(x)) - 
	\tanh(\tilde{\epsilon}_{\beta}(x))}\left(\frac{1}{\sqrt{1-x^2}\sqrt{1-y^2}}\right) \right| \leq \frac{\frac{2\delta}{\beta b} +2 \delta^2}{1-\delta^2}
\end{equation}
and therefore
\begin{equation}
	\left|f_{12} - f_{13}\right| \leq 2\left(\frac{2\delta^2 +\frac{\delta}{\beta b}}{1-\delta^2}\right)  ,
\end{equation}
repeating the argument for the $I_1$ integral yields the existence of some constant $C_{\alpha}>0$ such that
\begin{equation}
  \limsup_{\beta \rightarrow \infty} \frac{I_1}{\log(\beta)} \leq C_{\alpha} 2\left(\frac{2\delta^2}{1-\delta^2}\right).
\end{equation}
Since this is true for all $\delta \geq 0$, we have
\begin{equation}
  \limsup_{\beta \rightarrow \infty} \frac{I_1}{\log(\beta)} = 0,
\end{equation}
finishing Step 4. 
\end{itemize}

Therefore for $\left|a\right| < 2b$, by Step 4 and Eq. (\ref{eq:muualinformationbyI1I2}),
\begin{eqnarray}
\limsup_{\beta \rightarrow \infty} \frac{I_{\alpha,\infty}(A:B)}{\log(\beta)} & \leq & \limsup_{\beta \rightarrow \infty} \frac{2 J_{11}}{\pi^2 \log(\beta)}\\ 
&\leq& \frac{2\alpha}{\pi^2 \log(2)\left(\min\left\{1,\alpha\right\}\right)^2}  \min\left\{\frac{27}{2  \left(\min\left\{1,\alpha\right\}\right)},\frac{4}{\left|1-\alpha\right|}\right\},
\nonumber
\end{eqnarray}
completing the proof of the low-temperature limit.
\qed
\end{proof}

\end{document}